\documentclass[11pt]{amsart}

\usepackage{amsfonts}
\usepackage{mathrsfs}
\usepackage{mathtools}
\usepackage{amssymb}
\usepackage{latexsym,amsthm}
\usepackage[all]{xy}
\usepackage{graphicx}
\usepackage{wrapfig}
\usepackage{array,multirow}
\usepackage{bm}
\usepackage{enumitem,caption}

\usepackage{longtable}

\usepackage{wasysym}

\usepackage{hyperref}
\hypersetup{%
    pdfborder = {0 0 0},
    colorlinks,
    citecolor=red!50!black,
    filecolor=Darkgreen,
    linkcolor=blue!40!black,
    urlcolor=cyan!50!black!90
}

\usepackage{tikz,tikz-cd}
\usetikzlibrary{matrix,arrows,shapes,snakes}

\newcommand{\nc}{\newcommand}
\newcommand{\rnc}{\renewcommand}

\setlist[itemize,1]{leftmargin=.4in}
\setlist[enumerate,1]{leftmargin=.4in,label=(\roman*)}
\setlist[description,1]{leftmargin=.4in,font=\normalfont\itshape}

\allowdisplaybreaks[1]

\nc{\qq}{\qquad}
\nc{\qu}{\quad}

\newtheorem{thrm}{Theorem}[section]
\newtheorem{prop}[thrm]{Proposition}

\newtheorem{lemma}[thrm]{Lemma}
	
\newtheorem{conj}[thrm]{Conjecture}

\theoremstyle{definition}
\newtheorem{defn}[thrm]{Definition}

\theoremstyle{remark}
\newtheorem{rmk}[thrm]{Remark}

\nc{\rmkend}{\ensuremath{\diameter}}
\nc{\examend}{\ensuremath{\diameter}}
\nc{\defnend}{\ensuremath{\diameter}}

\numberwithin{equation}{section}

\nc{\eq}[1]{\begin{equation} #1 \end{equation}}
\nc{\eqrefs}[2]{\text{(\ref{#1}-\ref{#2})}}

\rnc\appendixname{}

\nc{\al}{\alpha}
\nc{\be}{\beta}
\nc{\eps}{\epsilon}
\nc{\veps}{\varepsilon}
\nc{\ga}{\gamma}
\nc{\Ga}{\Gamma}
\nc{\del}{\delta}
\nc{\Del}{\Delta}
\nc{\ze}{\zeta}
\nc{\ka}{\kappa}
\nc{\la}{\lambda}
\nc{\La}{\Lambda}
\nc{\si}{\sigma}
\nc{\Si}{\Sigma}
\nc{\vphi}{\varphi}
\nc{\om}{\omega}
\nc{\Om}{\Omega}

\nc{\A}{\mathbb{A}}
\nc{\C}{\mathbb{C}}
\nc{\F}{\mathbb{F}}
\nc{\K}{\mathbb{K}}
\nc{\N}{\mathbb{N}}
\nc{\Q}{\mathbb{Q}}
\nc{\R}{\mathbb{R}}
\nc{\Z}{\mathbb{Z}}

\nc{\mfsl}{\mathfrak{s}\mathfrak{l}}
\nc{\mfb}{\mathfrak{b}}
\nc{\mfg}{\mathfrak{g}}
\nc{\mfh}{\mathfrak{h}}
\nc{\mfk}{\mathfrak{k}}
\nc{\mfn}{\mathfrak{n}}

\nc{\ot}{\otimes}

\rnc{\t}{\mathrm{t}}
\nc{\e}{\mathrm{e}}
\nc{\id}{\mathrm{id}}
\nc{\Id}{\mathrm{Id}}

\let\Im\relax

\DeclareMathOperator{\End}{End}
\DeclareMathOperator{\GL}{GL}
\DeclareMathOperator{\Hom}{Hom}
\DeclareMathOperator{\Im}{Im}
\DeclareMathOperator{\Ker}{Ker}
\DeclareMathOperator*{\Tr}{Tr}

\nc{\wb}{\overline}
\nc{\wh}{\widehat}
\nc{\wt}{\widetilde}

\nc{\mc}{\mathcal}
\nc{\mf}{\mathfrak}

\nc{\red}{\color{red}}
\nc{\blu}{\color{blue}}
\nc{\br}{\color{brown}}
\nc{\grn}{\color{green!55!black}}
\nc{\gry}{\color{gray}}

\nc{\ben}{\begin{eqnarray*}}
\nc{\een}{\end{eqnarray*}}
\nc{\qhyp}[4]{ {}_{2} \phi_{1} \left( {\genfrac{}{}{0pt}{0}{#1}{#2}};{#3},{#4}\right)}
\nc{\txi}{\tilde \xi}
\nc{\adag}{{a^\dag}}
\nc{\balpha}{\bm \al}
\nc{\bbeta}{\bm \be}

\begin{document}

\title{A Q-operator for open spin chains I:  Baxter's TQ relation}

\author{Bart Vlaar}
\address{
Department of Mathematics, Heriot-Watt University, Edinburgh, EH14 4AS, UK;  The Maxwell Institute for Mathematical Sciences, Edinburgh}
\email{b.vlaar@hw.ac.uk, r.a.weston@hw.ac.uk}
\author{Robert Weston}

\begin{abstract} 
We construct a Q-operator for the open XXZ Heisenberg quantum spin chain with diagonal boundary conditions and give a rigorous derivation of Baxter's TQ relation.
Key roles in the theory are played by a particular infinite-dimensional solution of the reflection equation and by short exact sequences of intertwiners of the standard Borel subalgebras of $U_q(\wh{\mfsl}_2)$.
The resulting Bethe equations are the same as those arising from Sklyanin's algebraic Bethe ansatz.
\end{abstract}

	
\maketitle

\setcounter{tocdepth}{1} 

\tableofcontents


\newcommand{\ws}{\hspace*{3mm}}
\newcommand{\uq}{U_q(\widehat{\mathfrak{sl}_2})}
\newcommand{\uqn}{U_q(\widehat{\mathfrak{sl}_n})}
\newcommand{\uqg}{U_q(\mathfrak{g})}
\newcommand{\uqbpm}{U_q(\mathfrak{b}^\pm) }

\section{Introduction}

Baxter introduced the Q-operator in \cite{Ba72,Ba73a} as a device that enabled him to obtain Bethe equations for the eigenvalues of the 8-vertex model in the absence of a Bethe ansatz for the eigenvectors.
The general argument is as follows: let $\mc T(z)$ denote a transfer matrix which is diagonalizable, satisfies $[\mc T(y),\mc T(z)]=0$ for all $y,z \in \C$ and whose eigenvalues depend holomorphically on $z\in \C$.
Suppose that there exists a diagonalizable operator $\mc Q(z)$ which satisfies 
\eq{ \label{intro:commute}
[\mc T(y),\mc Q(z)]=[\mc Q(y),\mc Q(z)]=0
} 
for all $y,z \in \C$, and which has eigenvalues that are polynomial in $z$.
Suppose further that $\mc T(z)$ and $\mc Q(z)$ satisfy a relation of the form
\eq{
\mc T(z) \mc Q(z)= \alpha_+(z) \mc Q(q z) + \alpha_-(z) \mc Q(q^{-1}z),\label{eq:TQ1} 
}
for some meromorphic $\alpha_\pm$ and nonzero complex number $q$.
Now consider a simultaneous eigenvector of $\mc T(z)$ and $\mc Q(z)$ and denote the corresponding eigenvalues by $T(z)$ and $Q(z)$.
Then if $\{z_1, z_2,\cdots,z_M\}$ denote the roots of the polynomial $Q(z)$, holomorphicity of $T(z)$ requires that 
\eq{
T(z_i)Q(z_i)=0=\alpha_+(z_i) Q(q z_i) + \alpha_-(z) Q(q^{-1} z_i),\quad \hbox{for all} \ws i \in\{1,2, \cdots,M\}.\label{eq:BE1}
}
The right-hand equalities in \eqref{eq:BE1} constitute the Bethe equations for the set of roots.

This route to Bethe equations led to the solution of the 8-vertex model - clearly a major success.
However, the Q-operator subsequently fell out of favour for two reasons: firstly, for the 8-vertex model, Baxter himself made the Q-operator redundant by deploying the vertex-face correspondence in order to construct eigenvectors \cite{Ba73b}.
Secondly the method of constructing the Q-operator in \cite{Ba72} was ingenious, but also complicated and not obviously applicable to other models.

In order to generalize the 8-vertex Q-operator construction, what was lacking was an understanding of how the Q-operator fitted in with the Quantum Inverse Scattering Method/Quantum Groups picture of solvable lattice models.
This understanding was supplied more than twenty years later in the works of Bazhanov, Lukyanov and Zamolodchikov (BLZ) \cite{BLZ96,BLZ97,BLZ99}.
Note that $\mc T(z)$ can be constructed as the trace of the monodromy matrix over a finite-dimensional auxiliary space representation of the quantum group; the key idea of BLZ was that $\mc Q(z)$ can be constructed as a trace of the monodromy matrix over an infinite-dimensional representation of the (standard) Borel subalgebra of the quantum group.
The TQ relation \eqref{eq:TQ1} then arise due to fusion of the finite and infinite dimensional representations.
In more algebraic language, \eqref{eq:TQ1} arises directly from a short exact sequence of representations of the Borel subalgebra (for the XXZ model, this short exact sequence is the one given by Lemma \ref{lem:iotatau:intertwine} of the current paper).
The papers of BLZ are primarily concerned with conformal field theory.
Useful papers dealing with the same construction for lattice models and spin chains are \cite{AF97,Ko05}.
Alternative modern approaches to the Q-operator are given in, for instance, \cite{PG92,De99,Pr00}.

Most early work on the representation-theoretic construction of the Q-operator concerned the simplest case of the algebra $\uq$, or equivalently the XXZ quantum spin chain.
The reason for this limitation is the complexity of dealing with the required infinite-dimensional representations of the Borel subalgebras.
For $\uq$, these representations can be expressed simply using q-oscillators.
This q-oscillator approach has been generalized to $\uqn$ in \cite{BHKh02,Ko08}.
In order to go beyond $\uqn$ a more general construction of `asymptotic representations' of  $\uqg$ (the general quantum loop algebra of non-twisted type) was developed in \cite{HJ12}.
The central idea behind asymptotic representations is described succinctly in the paper \cite{BGKNR13} which deals with the $\uq$ case: a limit of the Verma module is taken with the associated weight going to infinity.
It is only possible to obtain a well-defined action of one of the Borel subalgebras of $\uq$ - not both - in this limit.
The Q-operators for $\uqg$ associated with these asymptotic representations were constructed in \cite{FH15} and their general TQ relations and Bethe equations were found.
The authors of \cite{FH15} thus proved a conjecture of Frenkel and Reshetikhin \cite{FR98} regarding the connection of Bethe equations to the q-characters of $\uqg$.
For recent developments in this direction, see \cite{HL16,FJMM17,He19}.\\

Everything we have discussed above concerns closed quantum spin chains, that is, those with periodic or quasi-periodic boundary condition in which the transfer matrix is given as a twisted trace of the monodromy matrix.
A general method for constructing integrable quantum spin chains with open quantum spin chains was developed in \cite{Sk88}.
Crucial in this method is the notion of the reflection equation (boundary Yang-Baxter equation), see \cite{Ch84}, solutions of which are known as K-matrices.
This leads to an integrable quantum spin chain described by a double-row transfer matrix which in our notation is given in \eqref{T:def}.
The choice of a solution of the reflection equation corresponds to the choice of a coideal subalgebra of the underlying quantum group \cite{KS92,DM03}.
Much work has been done to classify solutions to reflection equations and their associated coideal subalgebras (see \cite{RV16} and contained references).
In addition, there has recently been considerable progress in the construction of universal K-matrices analogous to the universal R-matrix for quantum groups of finite type, see \cite{BW13,BK16,Ko20}.

Sklyanin developed the algebraic Bethe ansatz for the example of the open XXZ model with diagonal boundary conditions in the founding paper \cite{Sk88} (the formulation is presented in the notation of the current paper in Appendix C).
Again, there has been a lot of effort since then to produce Bethe equations for large classes of other open integrable quantum spin chains.
In analogy with the situation for closed chains, there are two main approaches: the first produces TQ relations on the level of eigenvalues by fusing transfer matrices associated with finite-dimensional auxiliary spaces and then taking a formal limit with respect to the dimension of the auxiliary space \cite{YNZh06}; the second, and more widely used, is the analytic Bethe ansatz method \cite{AMN95}.
Connections to separation of variables techniques were made in \cite{DKM03, KMN14}.
Finally, there has been work on open Q-operators in the context of Markov matrices for asymmetric simple exclusion processes, see \cite{LP14}.

There is a some existing work on the explicit construction of open Q-operators.
Frassek and Sz\'ecs\'enyi \cite{FSz15} considered the XXX case with open boundary conditions; our work is similar in spirit to theirs however we have placed more emphasis on the representation-theoretic role of various linear maps.
Another closely related paper is that of Baseilhac and Tsuboi \cite{BTs18} who consider the same XXZ model with the same boundary conditions, but take the asymptotic approach: namely they construct monodromy matrices and K-matrices associated with an infinite-dimensional auxiliary space by taking the asymptotic limit of the Verma module.
In the paper \cite{BTs18} explicit formulas for families of Q-operator are given, although Baxter's relation \eqref{eq:TQ1} and the Bethe ansatz equations are not derived.\\

In this paper we give an algebraic construction of the Q-operator for the open XXZ model with diagonal boundary conditions. 
The idea is straightforward: we construct the Q-operator as a trace of the double-row monodromy matrix over infinite-dimensional representations of the Borel subalgebra as for the closed spin chain.
The execution is complicated by two issues: the first is that we must construct infinite-dimensional K-matrix solutions of the reflection equation.
In fact this is comparatively simple in the case of diagonal boundary conditions and the solution is given by Proposition \ref{prop:KW:def}.
The second issue is that there are now five algebras in play: the quantum group $\uq$; the two Borel subalgebras $\uqbpm$ (each associated with one of the rows of the double-row transfer matrix); and a coideal subalgebra for each boundary.
The convergent open Q-operator \eqref{Q:def} constructed in this paper does not need the regularizing twist that is required in the closed case.
The TQ relation is proven in Theorem \eqref{thm:QT:rel} using the short exact sequences \eqref{SES:plus} and \eqref{SES:min} for the two different Borel subalgebras along with the compatibility of the associated homomorphisms with the boundary fusion relations given by Lemma \ref{KW:iotatau:lem}. 
For the use of short exact sequences in deriving Baxter's TQ relation, also see \cite{AF97,RW02,Ko03,Ko05,JMS13}.

The finite-dimensional counterpart of these relations is given in \cite[Eq.~(4.7)]{RSV16} which in itself goes back to the original notion of K-matrix fusion in \cite{KS92,MN92}.
The boundary fusion relations, which we prove here using linear algebra, are the only key relations in this paper for which we do not have a full representation-theoretic understanding.
In addition we prove the important commutativity statement $[\mc T(y),\mc Q(z)]=0$ in Theorem \ref{thm:commute}.

In Section \ref{sec:Betheeqns} we show that the Q-operator satisfies a simple formula given in Theorem \ref{thm:Q:crossing}, known as \emph{crossing symmetry}, and use it, combined with the TQ relation, to recover the same Bethe equations as in the algebraic Bethe ansatz approach of Sklyanin \cite{Sk88} as reproduced in our Appendix \ref{app:ABA}.
The crossing symmetry property and hence the derivation of the Bethe equations rely on Conjecture \ref{conj:Q}, which states that the family $\{ \mc Q(y) \}$ is commutative and the matrix entries of the Q-operator are polynomial in $z$;  we cannot yet prove this conjecture in generality, but we do prove polynomiality for all diagonal entries in Appendix \ref{sec:Q:diagonalentries} and verify the full conjecture for two lattice sites in Appendix \ref{sec:lowN}.
Finally, for the purpose of comparison, we provide the analogous derivation of the TQ relation in the closed (twisted-periodic) case in Appendix \ref{app:closed}.

\subsection*{Acknowledgements}

This research is supported by EPSRC grant EP/R009465/1.
The authors would like to thank Pascal Baseilhac, Alexander Cooper, Anastasia Doikou, Heiko Gimperlein, David Hernandez, Christian Korff, Vidas Regelskis, Zengo Tsuboi and Paul Zinn-Justin for useful discussions and comments.


\section{Level-0 representation theory of $U_q(\widehat{\mathfrak{sl}}_2)$}

\subsection{Quantum affine $\mfsl_2$}

We denote by $\mfg$ the (derived) affine Lie algebra $\wh\mfsl_2$ and by $\si$ the associated nontrivial diagram automorphism, i.e.~the permutation of the set $\{0,1\}$.
For $p \in \C^\times := \C \backslash \{ 0 \}$ and elements $x,y$ of any algebra\footnote{
By algebra (except in the case of a Lie algebra) we always mean a unital associative algebra.
All (Lie) algebras in this paper are defined over $\C$.} 
we denote $[x,y]_p = x y - p y x$.
In this section, and in particular in the following definition, we allow $q \in \C \backslash \{ -1,0,1\}$ (later on we will restrict $q$).
The (derived) affine quantum group $U_q(\mfg)$ is the algebra \cite{Dr85,Dr86,Ji86} with generators $e_i, f_i, k_i^{\pm 1}$ ($i\in \{0,1\}$) and relations
\begin{gather}
k_i e_i = q^2 e_i k_i, \qq k_i f_i = q^{-2} f_i k_i, \qq [e_i,f_i]_1 = \frac{k_i-k_i^{-1}}{q-q^{-1}},\label{Uq:relations1} \\
k_i e_{\si(i)} = q^{-2} e_{\si(i)} k_i, \qq k_i f_{\si(i)} = q^2 f_{\si(i)} k_i, \qq [k_i,k_{\si(i)}]_1=[e_i,f_{\si(i)}]_1 = 0, \label{Uq:relations2} \\
[e_i,[e_i,[e_i,e_{\si(i)}]_{q^2}]_1]_{q^{-2}} = [f_i,[f_i,[f_i,f_{\si(i)}]_{q^2}]_1]_{q^{-2}} = 0, \label{Uq:relations3}
\end{gather}
for all $i \in \{0,1\}$.
The following assignments define a coproduct on $U_q(\mfg)$:
\eq{ \label{Uq:bialgebra}
\begin{aligned}
\Del(e_i) &= e_i \ot 1 + k_i \ot e_i \qu & \Del(f_i) &= f_i \ot k_i^{-1} + 1 \ot f_i \qu & \Del(k_i^{\pm 1}) &= k_i^{\pm 1} \ot k_i^{\pm 1} 
\end{aligned}
}
which restricts to a coproduct on the following important subalgebras:
\[ 
U_q(\mfh) := \langle k_0^{\pm 1}, k_1^{\pm 1} \rangle, \qq U_q(\mfb^+) := \langle e_0,e_1,k_0^{\pm 1},k_1^{\pm 1} \rangle, \qq U_q(\mfb^-) := \langle f_0, f_1, k_0^{\pm 1}, k_1^{\pm 1} \rangle.
\]
Equivalently, we can define $U_q(\mfb^+)$ and $U_q(\mfb^-)$ as the algebras generated by $e_0,e_1,k_0^{\pm 1},k_1^{\pm 1}$ and $f_0,f_1,k_0^{\pm 1},k_1^{\pm 1}$ respectively, subject to only those relations among \eqrefs{Uq:relations1}{Uq:relations3} which only contain those symbols; the respective identifications of generators extend to algebra embeddings $U_q(\mfb^\pm) \to U_q(\mfg)$.
This is important since we will be looking at representations of $U_q(\mfb^+)$ and $U_q(\mfb^-)$.

\subsection{Level-0 representations of quantum affine $\mfsl_2$ and its Borel subalgebras}

Note that $k_0k_1$ is central in $U_q(\mfg)$, and  for any representation $\rho$ of $U_q(\mfg)$, $U_q(\mfb^+)$ or $U_q(\mfb^-)$, the \emph{level} $\la$ is the unique complex number such that $\rho(k_0k_1)$ acts on the image of $\rho$ as the scalar $q^\la$.
Let $I$ be the two-sided ideal of $U_q(\mfg)$ generated by $k_0k_1-1$.
If a representation has level 0 then it factors through a representation of the quotient $U_q(\mfg)/I$, which is isomorphic to the \emph{quantum loop algebra} of $\mfsl_2$.
In particular, all finite-dimensional representations of $U_q(\mfg)$ have level 0.
We will restrict our attention to level-0 representations from now on; on the other hand we widen the scope somewhat by also including representations of the subalgebras $U_q(\mfb^\pm)$ where $k_0k_1$ acts as the identity.
For a more comprehensive discussion about level-0 representations of $U_q(\mfg)$ as evaluation modules constructed from Verma modules for $U_q(\mfsl_2)$ see for example \cite[Section 3]{BGKNR13} and \cite[Section 2]{RSV16}.

All vector spaces under consideration and their tensor products are defined over $\C$.
Let $U,U'$ be any two vector spaces.
We shall write $\Hom(U,U')$ to mean the set of $\C$-linear maps from $U$ to $U'$ and $\End(U) = \Hom(U,U)$.
We denote by $P^{U,U'} \in \Hom(U \ot U',U' \ot U)$ the map that na\"{i}vely swaps tensor factors: $P^{U,U'}(u \ot u') = u' \ot u$ for all $u \in U$ and $u' \in U'$.
If there is no cause for confusion, we will simply write $P$.
In this paper we will focus on two vector spaces in particular and tensor products of them.

Let $z \in \C^\times$.
Also, let $V = \C^2$ and choose any ordered basis $(v^0,v^1)$ of $V$.
Expressing linear operators on $V$ as $2 \times 2$ matrices with respect to this basis, the following assignments extend to an algebra homomorphism $\pi_z: U_q(\mfg) \to \End(V)$ (a level-0 representation of $U_q(\mfg)$ on $V$).
\begin{align*}
\pi_z(e_0) &= \begin{pmatrix} 0 & 0 \\ z & 0 \end{pmatrix} & 
\pi_z(f_0) &= \begin{pmatrix} 0 & z^{-1} \\ 0 & 0 \end{pmatrix} & 
\pi_z(k_0) &= \begin{pmatrix} q^{-1} & 0 \\ 0 & q \end{pmatrix} \\
\pi_z(e_1) &= \begin{pmatrix} 0 & z \\ 0 & 0 \end{pmatrix} & 
\pi_z(f_1) &= \begin{pmatrix} 0 & 0 \\ z^{-1} & 0 \end{pmatrix} & 
\pi_z(k_1) &= \begin{pmatrix} q & 0 \\ 0 & q^{-1} \end{pmatrix}.
\end{align*}

Consider the infinite-dimensional vector space
\eq{ \label{W:def}
W = \bigoplus_{j =0}^\infty \C w^j
}
and let $\mc F$ denote the commutative algebra of functions$: \Z_{\ge 0} \to \C$.
We define certain linear maps on $W$ as follows:
\eq{
\label{oscq:rep} a(w^{j+1}) = w^j, \qq a^\dag(w^j) = (1-q^{2(j+1)}) w^{j+1}, \qq f(D)(w^j) = f(j) w^j 
}
for all $j \in \Z_{\ge 0}$ and $f \in \mc F$.
The linear maps $a,a^\dag$ and $f(D)$ for all $f \in \mc F$ satisfy the defining relations of (an extension of) the q-oscillator algebra, to wit
\eq{ \label{oscq:relations}
f(D) a = a f(D-1), \qq f(D) a^\dag = a^\dag f(D+1), \qq a^\dag a = 1-q^{2D}, \qq a a^\dag = 1-q^{2(D+1)} .
}
We denote the subalgebra of $\End(W)$ generated by $a,a^\dag$ and the commutative algebra $\mc F(D) := \{ f(D) \, | \, f \in \mc F \}$ by ${\rm osc}_q$.
As a consequence of \eqref{oscq:relations} we have the linear decomposition
\eq{ \label{oscq:decomposition}
{\rm osc}_q =  {\mc F}(D) \oplus \bigoplus_{m \in \Z_{\ge 1}} \Big( (a^\dag)^m {\mc F}(D) \oplus {\mc F}(D) a^m  \Big).
}
For more details on the q-oscillator algebra, the reader may consult e.g.~\cite{Ku91} and \cite[Sec.~5]{KSch12}.

\begin{prop} \label{prop:rhoplus:rep}
Let $r,z \in \C^\times$.
The following assignments define algebra homomorphisms $\rho^+_{z,r}: U_q(\mfb^+) \to \End(W)$ and $\rho^-_{z,r}: U_q(\mfb^-) \to \End(W)$:
\eq{ \label{rhoplus:def}
\begin{aligned}
\rho^+_{z,r}(e_0) &= \frac{q^{-1}z}{q-q^{-1}} a^\dag, \qq & 
\rho^+_{z,r}(k_0) &= r q^{2D}, \\
\rho^+_{z,r}(e_1) &= \frac{qz}{q-q^{-1}} a, &
\rho^+_{z,r}(k_1) &= r^{-1} q^{-2D}.
\end{aligned}
}
and
\eq{ \label{rhominus:def}
\begin{aligned}
\rho^-_{z,r}(f_0) &= \frac{qz^{-1}}{q-q^{-1}} a, \qq & 
\rho^-_{z,r}(k_0) &= r^{-1} q^{2D}, \\
\rho^-_{z,r}(f_1) &= \frac{q^{-1}z^{-1}}{q-q^{-1}} a^\dag, & 
\rho^-_{z,r}(k_1) &= r q^{-2D}.
\end{aligned}
}
\end{prop}

\begin{rmk} \mbox{}
\begin{enumerate}
\item These two families of representations depend on an additional parameter $r$ which does not appear in representations of $U_q(\wh \mfsl_2)$.
We need this parameter in order to have a short exact sequence of $U_q(\mfb^+)$-intertwiners, see Lemma \ref{lem:iotatau:intertwine}.
The Q-operator does not depend on $r$.
\item In order to define the representations $\rho^\pm_{z,r}$ one does not need to work with the whole of $\mc F(D)$, but only by the subalgebra generated by $q^{2D}$ and $q^{-2D}$.
However for the various intertwiners we will need a larger algebra; namely for infinite-dimensional solutions of the Yang-Baxter equation (L-operators) we need the subalgebra of $\mc F(D)$ generated by $q^{D}$ and $q^{-D}$, see \eqref{L:def}.
For infinite-dimensional solutions of the reflection equation it is convenient to allow all of $\mc F(D)$, see \eqref{KW:def}.
\hfill \rmkend
\end{enumerate}
\end{rmk}

\begin{proof}[Proof of Proposition \ref{prop:rhoplus:rep}]
We give the proof for $\rho^+_{z,r}$; the proof for $\rho^-_{z,r}$ is analogous.
We only need to verify the relations 
\[
k_i e_i = q^2 e_i k_i, \qq k_i e_{\si(i)} = q^{-2} e_{\si(i)} k_i, \qq [k_i,k_{\si(i)}]=0, \qq 
[e_i,[e_i,[e_i,e_{\si(i)}]_{q^2}]_1]_{q^{-2}} =0
\]
for $i \in \{ 0 ,1\}$.
Only the last is nontrivial; we have
\[
q  \rho^+_{z,r}([e_0,e_1]_{q^{-2}}) = - q^{-1} \rho^+_{z,r}([e_1,e_0]_{q^2}) = \frac{z^2}{q-q^{-1}} 
\]
and hence
\[
\rho^+_{z,r}([e_0,[e_0,e_1]_{q^{-2}}]_1) = \rho^+_{z,r}([e_1,[e_1,e_0]_{q^2}]_1) = 0,
\]
from which the result follows.
\end{proof}

It can be checked that the algebra maps $\rho^{\pm}_{z,r}$ with domain $U_q(\mfb^\pm)$ cannot be extended to algebra maps of $U_q(\mfg)$.

\subsection{R-matrices and their infinite-dimensional counterparts}

Let $\mc R$ be the parameter-independent universal R-matrix of $U_q(\mfg)$, see e.g.~\cite{KhT92,BGKNR10}.
It satisfies
\begin{align}
\label{uniR:1} \mc R \Del(u) &= \Del^{\rm op}(u) \mc R \qq \text{for all } u \in U_q(\mfg) \\
\label{uniR:2} (\Del \ot \id)(\mc R) &= \mc R_{13} \mc R_{23} \\
\label{uniR:3} (\id \ot \Del)(\mc R) &= \mc R_{13} \mc R_{12} 
\end{align}
and as a consequence
\eq{
\label{uniR:YBE} \mc R_{12} \mc R_{13} \mc R_{23} = \mc R_{23} \mc R_{13} \mc R_{12}.
}
From the quantum double construction it follows that $\mc R$ lies in a completion of $U_q(\mfb^+) \ot U_q(\mfb^-)$; in particular the following linear maps are well-defined:
\begin{gather*}
(\pi_{z_1} \ot \pi_{z_2})(\mc R) \in \End(V \ot V) \\
(\rho^+_{z_1,r} \ot \pi_{z_2})(\mc R) \in \End(W \ot V), \qq (\pi_{z_1} \ot \rho^-_{z_2,r^{-1}})(\mc R) \in \End(V \ot W).
\end{gather*}
These operators depend rationally on the quotient $z_1/z_2$, see e.g.~\cite[Lecture 9]{EFK98}.
We denote suitable scalar multiples of these operators by $R(z_1/z_2)$, $L(z_1/z_2,r)$ and $L^-(z_1/z_2,r)$, respectively.
In order to write down explicit expressions for these linear operators it is customary to apply the appropriate representation to \eqref{uniR:1}, restricting to a subalgebra where appropriate, and solve the resulting linear equations, which are
\begin{align*} 
R(z_1/z_2) (\pi_{z_1} \ot \pi_{z_2})(\Del(u)) &= (\pi_{z_1} \ot \pi_{z_2})(\Del^{\rm op}(u)) R(z_1/z_2) && \text{for all } u \in U_q(\mfg), \\
L^+(z_1/z_2,r) (\rho^+_{z_1,r} \ot \pi_{z_2})(\Del(u)) &= (\rho^+_{z_1,r} \ot \pi_{z_2})(\Del^{\rm op}(u)) L^+(z_1/z_2,r) && \text{for all } u \in U_q(\mfb^+), \\
L^-(z_1/z_2,r) (\pi_{z_1} \ot \rho^-_{z_2,r^{-1}})(\Del(u)) &= (\pi_{z_1} \ot \rho^-_{z_2,r^{-1}})(\Del^{\rm op}(u))  L^-(z_1/z_2,r) && \text{for all } u \in U_q(\mfb^-).
\end{align*}
This leads us to the following solutions (unique up to overall scalar multiples):
\eq{
\label{R:def}
R(z) = \begin{pmatrix} 1 - q^2 z^2 & 0 & 0 & 0 \\ 0 & q(1-z^2) & (1-q^2) z & 0 \\ 0 & (1-q^2) z & q(1-z^2) & 0 \\ 0 & 0 & 0 & 1 - q^2 z^2 \end{pmatrix} \in \End(V \ot V)
}
and 
\eq{ 
L^+(z,r) = L(z,r) \in \End(W \ot V), \qq L^-(z,r) = P L(z,r) P \in \End(V \ot W)
}
where
\eq{
\label{L:def} 
L(z,r) = \begin{pmatrix} 1 & -q^{-1} z a^\dag \\ - q z a & 1-q^{2(D+1)} z^2 \end{pmatrix} \begin{pmatrix} q^D r & 0 \\ 0 & q^{-D} \end{pmatrix} .
}
Note that $L(z,r)$ is invertible if $z^2 \ne 1$, with the inverse given by
\eq{
\label{L:inv} 
L(z,r)^{-1} = \frac{1}{1-z^2} \begin{pmatrix} q^{-D} r^{-1} & 0 \\ 0 & q^D \end{pmatrix} \begin{pmatrix} 1-q^{2D} z^2 & q^{-1} z a^\dag \\ q z a & 1 \end{pmatrix}
}

As a consequence of \eqref{uniR:YBE} we obtain the following \emph{Yang-Baxter equations}:
\begin{align}
\label{R:YBE} R_{12}(\tfrac{z_1}{z_2}) R_{13}(\tfrac{z_1}{z_3}) R_{23}(\tfrac{z_2}{z_3}) &= R_{23}(\tfrac{z_2}{z_3}) R_{13}(\tfrac{z_1}{z_3}) R_{12}(\tfrac{z_1}{z_2}) && \in \End(V \ot V \ot V) \\
\label{Lplus:YBE} L_{12}(\tfrac{z_1}{z_2},r) L_{13}(\tfrac{z_1}{z_3},r) R_{23}(\tfrac{z_2}{z_3}) &= R_{23}(\tfrac{z_2}{z_3}) L_{13}(\tfrac{z_1}{z_3},r) L_{12}(\tfrac{z_1}{z_2},r) && \in \End(W \ot V \ot V) \\
\label{Lmin:YBE} R_{12}(\tfrac{z_1}{z_2}) L_{13}(\tfrac{z_1}{z_3},r) L_{23}(\tfrac{z_2}{z_3},r) &= L_{23}(\tfrac{z_2}{z_3},r) L_{13}(\tfrac{z_1}{z_3},r) R_{12}(\tfrac{z_1}{z_2}) && \in \End(V \ot V \ot W)
\end{align}
which are understood as equations for operator-valued rational functions in $z_1,z_2,z_3$.
Here we have introduced the usual subscript notation.
Namely, for $X(z) \in \End(V_m \ot V_n)$ depending meromorphically on $z$, $N \in \Z_{>0}$ and $m,n \in \{1,\ldots,N\}$ with $m \ne n$, the notation $X_{mn}(z)$ for an element of $\End(V_1 \ot \cdots \ot V_N)$ with the first and second tensor factors of $X(z)$ act nontrivially on the $m$-th and $n$-th tensor factors of $V_1 \ot \cdots \ot V_N$, respectively
(here $N=3$).

We highlight the so-called crossing property of $R(z)$ and $L^\pm(z,r)$.
For convenience we use the abbreviations
\[
\wt R(z) := ((R(z)^{t_1})^{-1})^{t_1} = ((R(z)^{t_2})^{-1})^{t_2}, \qq \qq \wt L(z,r) := ((L(z,r)^{t_2})^{-1})^{t_2}
\]
in terms of partial transpositions $\cdot^{\t_{1,2}}$ with respect to the first and second factors in the tensor products $V \ot V$ and $W \ot V$ respectively.
Since $R(z)^{\t_1}$, $R(z)^{\t_2}$, $L(z,r)^{\t_2}$ depend polynomially on $z$ and are invertible for $z=0$ it follows that $\wt R(z)$ and $\wt L(z,r)$ are well-defined for all but finitely many values of $z$.
As a consequence of \cite[Thm.~4.1]{FR92}, $\wt R(z)^{-1}$ can be expressed in terms of $R(q^2 z)$.
In particular, in our case due to our choice of normalization of $R(z)$ we have 
\eq{ \label{R:cross.unit.}
\wt R(z)^{-1} = \frac{(1 - z^2)(1 - q^4 z^2)}{(1 - q^2 z^2)(1 - q^6 z^2)} R(q^2 z).
}
The operator $L(z,r)$ enjoys a similar identity; in fact, the proof of \cite[Thm.~4.1]{FR92}, which is stated only for tensor products of finite-dimensional representations of $U_q(\mfg)$, applies in this setting as well.
Indeed, straightforward computations give
\eq{ \label{L:T2}
\begin{aligned}
L(z,r)^{\t_2} &= \begin{pmatrix} q^{D} r & 0 \\ 0 & q^{-D} \end{pmatrix} \begin{pmatrix} 1 & -q^2 z a \\ -z a^\dag & 1-q^{2(D+1)}z^2 \end{pmatrix} \\
(L(z,r)^{\t_2})^{-1} &= \frac{1}{1-q^2z^2} \begin{pmatrix} 1-q^{2(D+2)}z^2 & q^2 z a \\ z a^\dag & 1 \end{pmatrix} \begin{pmatrix} q^{-D} r^{-1} & 0 \\ 0 & q^{D} \end{pmatrix} .
\end{aligned} 
}
Combining this with \eqref{L:inv} we readily obtain
\eq{ \label{L:cross.unit.}
\wt L(z,r)^{-1} = \frac{1 - q^2 z^2}{1 - q^4 z^2} L(q^2 z,r).
}

\subsection{K-matrices and their infinite-dimensional counterparts}

Let $\xi \in \C^\times$.
The matrix
\eq{ \label{KV:def}
K^V(z) := \begin{pmatrix} \xi z^2 - 1 & 0 \\ 0 & \xi - z^2 \end{pmatrix} \in \End(V)
}
is a solution of the \emph{(finite-finite) right reflection equation}
\eq{ \label{KV:RE}
R(y/z) K^V_1(y) R(y z) K^V_2(z) = K^V_2(z) R(y z) K^V_1(y) R(y/z) \in \End(V \ot V),
}
see \cite{Sk88}.
If we include the ``limit cases'' $\Big( \begin{smallmatrix} 1 & 0 \\ 0 & z^{\pm 2} \end{smallmatrix} \Big)$ then this family provides all invertible diagonal solutions $K^V(z) \in \End(V)$ of \eqref{KV:RE} up to an overall scalar factor, see e.g.~\cite[Lem. 2.2]{RSV15}.

\begin{rmk} \label{rmk:coidealsubalgebra}
Consider the unique involutive Lie algebra automorphism $\theta$ of $\mfg$ such that $e_i \leftrightarrow -f_{\si(i)}$ for $i \in \{0,1\}$ and let $\mfk = \mfg^\theta$ be the corresponding fixed-point subalgebra.
The following subalgebra of $U_q(\mfg)$ quantizes $U(\mfk) \subset U(\mfg)$:
\[
U_q(\mfk) := \langle e_0 - q^{-1} \xi k_0 f_1 , e_1 - q^{-1} \xi^{-1} k_1 f_0, k_0 k_1^{-1}, k_0^{-1} k_1 \rangle
\]
where $\xi \in \C^\times$ tends to 1 as $q \to 1$; see \cite{Ko14} for a more general theory of quantized fixed-point subalgebras of Kac-Moody Lie algebras.
Note that, unlike $U(\mfk)$, the coproduct on $U_q(\mfg)$ does not restrict to one on $U_q(\mfk)$.
Instead, $U_q(\mfk)$ is a \emph{left coideal}:
\eq{ \label{Uqk:coideal}
\Del(U_q(\mfk)) \subset U_q(\mfg) \ot U_q(\mfk).
}
Note that $U_q(\mfk)$ possesses an independent description as a $U_q(\mfg)$-comodule algebra (in terms of generators and relations) and in such a context is also known as the \emph{augmented q-Onsager algebra}, see \cite{BB13}.

The matrix $K^V(z) \in \GL(V)$ is a $U_q(\mfk)$-module homomorphism.
More precisely, it intertwines the representations $(V,\pi_z|_{U_q(\mfk)})$ and $(V,\pi_{1/z}|_{U_q(\mfk)})$, i.e.:
\eq{ \label{KV:intw}
K^V(z) \pi_z(x) = \pi_{1/z}(x) K^V(z), \qq \text{for all } x \in U_q(\mfk).
}
The simple linear relation \eqref{KV:intw} defines $K^V(z)$ uniquely, up to a scalar.
In general it is known that, under certain technical assumptions, such intertwiners satisfy the reflection equation \eqref{KV:RE}, see \cite{DM03,DG02} and \cite[Sec.~6.2]{RV16}. \hfill \rmkend
\end{rmk}

We will also consider an element $K^W(z,r) \in \End(W)$ which satisfies the \emph{(infinite-finite) right reflection equation}
\eq{ \label{KW:RE}
L(y/z,r) K^W_1(y,r) L(y z,r) K^V_2(z) = K^V_2(z) L(y z,r) K^W_1(y,r) L(y/z,r) \in \End(W \ot V).
}

\begin{rmk}
As opposed to $K^V(z)$, the map $K^W(z,r)$ does not satisfy a simple relation of the form \eqref{KV:intw} for the representations $\rho^\pm_{z,r}$, since $U_q(\mfk)$ is not contained in either $U_q(\mfb^+)$ or $U_q(\mfb^-)$.
\hfill \rmkend
\end{rmk}

We denote the $q^2$-deformed Pochhammer symbol by
\[
(x)_j := \begin{cases} 
\displaystyle\prod_{i=0}^{j-1} (1-q^{2i} x) & \text{if } j \ge 0 \\
\displaystyle\prod_{i=1}^{-j} (1-q^{-2i} x)^{-1} & \text{if } j < 0
\end{cases}
\]
where $x \in \C$ and $j \in \Z \cup \{ \infty \}$ - note that this depends holomorphically on $q$ in the unit disk and, for $j \ge 0$, holomorphically on $x$.

\begin{prop} \label{prop:KW:def}
The unique solution of \eqref{KW:RE} such that $K^W(z,r)(w^0)=w^0$ is given by
\eq{ \label{KW:def}
K^W(z,r) = q^{-D^2} r^{-D} (-\xi)^D (q^2z^2/\xi)_D.
}
That is,
\[
K^W(z,r)(w^j) = q^{-j^2} r^{-j} (-\xi)^j (q^2z^2/\xi)_j w^j = (q r^{-1})^j \left( \prod_{i=1}^j (z^2 - q^{-2i} \xi) \right) w^j.
\]
\end{prop}

\begin{proof}
It is convenient to re-write \eqref{KW:RE} in the form
\[
K^W_1(y,r) X^+(y,z,r) = X^-(y,z,r) K^W_1(y,r)
\]
where
\begin{align*}
& X^+(y,z,r) := \Big( 1-\tfrac{y^2}{z^2} \Big) L(y z,r) K^V_2(z) L(y/z,r)^{-1} \\
& \qu = \begin{pmatrix} 1 & -q^{-1} y z a^\dag \\ - q y z a & 1-q^{2(D+1)} y^2 z^2 \end{pmatrix} \begin{pmatrix} \xi z^2 - 1 & 0 \\ 0 & \xi - z^2 \end{pmatrix}  \begin{pmatrix} 1-q^{2D} \frac{y^2}{z^2} & q^{-1} \frac{y}{z} a^\dag \\ q \frac{y}{z} a & 1 \end{pmatrix} \\
& \qu = \begin{pmatrix} 
(z^2-y^2) \xi + y^2 z^2 - 1 + \frac{y^2}{z^2} (1-z^4) q^{2D} & -q^{-1} \frac{y}{z} (1-z^4) a^\dag \\ 
- q \frac{y}{z} (1-z^4) a(y^2 q^{2D} - \xi) & y^2 - z^2 + (1- y^2 z^2) \xi - y^2 (1-z^4) q^{2(D+1)} \end{pmatrix}
\\
& X^-(y,z,r) := \Big( 1-\tfrac{y^2}{z^2} \Big)  L(y/z,r)^{-1} K^V_2(z) L(y z,r) \\
& \qu = \begin{pmatrix} 1-q^{2D} \frac{y^2}{z^2} & q^{-1} \frac{y}{z} a^\dag \\ q \frac{y}{z} a & 1 \end{pmatrix} \begin{pmatrix} \xi z^2 - 1 & 0 \\ 0 & \xi - z^2 \end{pmatrix} \begin{pmatrix} 1 & -q^{-1} y z a^\dag \\ - q y z a & 1-q^{2(D+1)} y^2 z^2 \end{pmatrix} \\
& \qu = \begin{pmatrix} 
(z^2-y^2) \xi + y^2 z^2 - 1 + \frac{y^2}{z^2} (1-z^4) q^{2D} & - r^{-1} \frac{y}{z} (1-z^4) (y^2 - \xi q^{-2D}) a^\dag \\ 
-r \frac{y}{z} (1-z^4) a q^{2D} & y^2 - z^2 + (1- y^2 z^2) \xi - y^2 (1-z^4) q^{2(D+1)}
\end{pmatrix}
\end{align*}
by virtue of \eqrefs{L:def}{L:inv}.
Considering that the diagonal entries are identical and how the off-diagonal entries are related we obtain the result.
\end{proof}

\begin{rmk}
This K-matrix is part of a larger family of infinite-dimensional solutions of various reflection equations found in \cite[Sec.~4]{BTs18}.
More precisely, \cite[Eq. (4.53)]{BTs18} is equivalent to \eqref{KW:RE} and the solution \eqref{KW:def} corresponds to \cite[Eq. (4.37)]{BTs18} up to an overall scalar.
We thank P.~Baseilhac and Z.~Tsuboi for pointing this out.
\hfill \rmkend
\end{rmk}

We will also be interested in solutions $\wt K^V(z) \in \End(V)$ and $\wt K^W(z,r) \in \End(W)$ of the \emph{left reflection equations}
\begin{align}
\label{KtV:RE}
\wt K^V_2(z) \wt R(y z) \wt K^V_1(y) R(\tfrac{y}{z})^{-1} &= R(\tfrac{y}{z})^{-1} \wt K^V_1(y) \wt R(y z) \wt K^V_2(z) && \in \End(V \ot V), \hspace{-4mm} \\
\label{KtW:RE}
\wt K^V_2(z) \wt L(y z,r) \wt K^W_1(y,r) L(\tfrac{y}{z},r)^{-1} &= L(\tfrac{y}{z},r)^{-1} \wt K^W_1(y,r) \wt L(y z,r) \wt K^V_2(z) \hspace{-4mm} && \in \End(W \ot V).
\hspace{-4mm}
\end{align}

Consider \eqref{KtW:RE}.
Employing \eqref{L:cross.unit.}, inverting and reparametrizing $(y,z) \mapsto (q^{-1} y,q^{-1} z)$, we see that it is equivalent to 
\eq{ \label{KtW:RE:alt}
L(\tfrac{y}{z},r) \wt K^W_1(q^{-1} y,r)^{-1} L(y z,r) \wt K^V_2(q^{-1} z)^{-1} = \wt K^V_2(q^{-1} z)^{-1} L(y z,r) \wt K^W_1(q^{-1} y,r)^{-1} L(\tfrac{y}{z},r).
}
Let $\tilde \xi \in \C^\times$ be arbitrary.
Comparing \eqref{KtW:RE:alt} with \eqref{KW:RE} we obtain that 
\eq{ \label{Kt:fromK}
(\wt K^V(z),\wt K^W(z,r)) := (f^V(z) K^V(q z)^{-1},f^W(z) K^W(q z,r)^{-1})|_{\xi \to \tilde \xi^{-1}}
}
satisfies \eqref{KtW:RE} with $K^V$ and $K^W$ given by \eqref{KV:def} and \eqref{KW:def}.
Here $f^V$, $f^W$ are any scalars depending meromorphically on $z$ and $\tilde \xi \in \C^\times$ is arbitrary.
By choosing $f^V(z) = (q^2 \tilde \xi z^2 - 1)(q^2 z^2 - \tilde \xi)$, $f^W(z) = (1 - q^2 \tilde \xi z^2)^{-1}$ we obtain
\begin{align}
\label{KtV:def} \wt K^V(z) &= \begin{pmatrix} q^2 \tilde \xi z^2 - 1 & 0 \\ 0 & \tilde \xi - q^2 z^2 \end{pmatrix} \\
\label{KtW:def} \wt K^W(z,r) &= (1 - q^2 \tilde \xi z^2)^{-1} q^{D^2} r^D  (-\tilde \xi)^D (q^4 \tilde \xi z^2)_D^{-1} = q^{D^2} r^D (-\tilde \xi)^D (q^2 \tilde \xi z^2)_{D+1}^{-1}.
\end{align}
Following the same argument, one obtains that \eqref{KtV:RE} is satisfied.


\section{Short Exact Sequences and Fusion}

In this section we construct $U_q(\mfb^+)$-intertwiners which take part in short exact sequences relating the module $(W \ot V,\rho^+_{z,r} \ot \pi_z)$ to the module $(W,\rho^+_{z',r'})$ for certain shifted parameters $z',r'$.
Similarly, we construct $U_q(\mfb^-)$-intertwiners involved in short exact sequences relating the $U_q(\mfb^-)$-modules $(V \ot W,\pi_z \ot \rho^-_{z,r})$ and $(W,\rho^-_{z'',r''})$ for certain shifted parameters $z'',r''$.

For $r \in \C^\times$, consider the following linear maps called \emph{fusion intertwiners}:
\eq{ \label{iotatau:formula}
\iota(r) :=  \begin{pmatrix} q^{-D} a^\dag \\ q^{D+1} r \end{pmatrix} \in \End(W,W \ot V), \qq 
\tau(r) := \big( q^D , \,  - q^{-D} r^{-1} a^\dag \big) \in \End(W \ot V,W).
}
Here we interpret elements of $W \ot V$ and elements of $V \ot W$ as vectors with two entries in $W$, using the ordered basis $(v^0,v^1)$ of $V$.
In other words, we have
\begin{align*}
\iota(r)(w^j) &= (q^{-j-1} - q^{j+1}) w^{j+1} \ot v^0 + q^{j+1} r w^j \ot v^1 \hspace{-60mm} \\
\tau(r)(w^j \ot v^0) &= q^j w^j & \tau(r)(w^j \ot v^1) &= (q^{j+1}-q^{-j-1}) r^{-1} w^{j+1} 
\end{align*}

\begin{lemma} \label{lem:iotatau:intertwine}
The maps $\iota(r)$ and $\tau(r)$ are $U_q(\mfb^+)$-intertwiners as follows: 
\eq{ \label{iotatau:intertwine:plus}
\iota(r): (W,\rho^+_{qz,qr}) \to (W \ot V,\rho^+_{z,r} \ot \pi_z), 
\qq \tau(r): (W \ot V,\rho^+_{z,r} \ot \pi_z) \to (W,\rho^+_{q^{-1} z,q^{-1} r})
}
and we have the following short exact sequence of $U_q(\mfb^+)$-intertwiners:
\begin{equation} \label{SES:plus}
\begin{tikzcd}
0 \arrow[r] & (W,\rho^+_{qz,qr}) \arrow[r,"\iota(r)"] & (W \ot V,\rho^+_{z,r} \ot \pi_z) \arrow[r,"\tau(r)"]  & (W,\rho^+_{q^{-1}z,q^{-1}r}) \arrow[r] & 0 
\end{tikzcd}
\end{equation}
Similarly, the maps $P \iota(r^{-1})$ and $\tau(r^{-1}) P$ are $U_q(\mfb^-)$-intertwiners:
\eq{ \label{iotatau:intertwine:min}
\begin{aligned}
P \iota(r^{-1}): & \; (W,\rho^-_{q^{-1} z,q^{-1} r}) \to (V \ot W,\pi_z \ot \rho^-_{z,r}), \\
\tau(r^{-1})P: & \; (V \ot W,\pi_z \ot \rho^-_{z,r}) \to (W,\rho^-_{q z,q r})
\end{aligned}
}
and we have the following short exact sequence of $U_q(\mfb^-)$-intertwiners:
\begin{equation} \label{SES:min}
\begin{tikzcd}
0 & (W,\rho^-_{qz,qr}) \arrow[l] & (V \ot W,\pi_z \ot \rho^-_{z,r}) \arrow[l,"\tau(r^{-1})P"] & (W,\rho^-_{q^{-1}z,q^{-1}r}) \arrow[l,"P\iota(r^{-1})"]  & 0 \arrow[l].
\end{tikzcd}
\end{equation}
\end{lemma}

\begin{proof}
The requirements \eqref{iotatau:intertwine:plus} and \eqref{iotatau:intertwine:min} that these maps are intertwiners between the indicated modules is equivalent to
\begin{align}
\label{iotatauplus:intertwine}
\begin{aligned}
\iota(r) \circ \rho^+_{qz,qr}(u) &= (\rho^+_{z,r} \ot \pi_z)(\Del(u)) \circ \iota(r) \\ 
\rho^+_{q^{-1} z, q^{-1} r}(u) \circ  \tau(r) &= \tau(r) \circ (\rho^+_{z,r} \ot \pi_z)(\Del(u))
\end{aligned} 
\Bigg\} & \qu \text{for all } u \in U_q(\mfb^+) \\
\label{iotataumin:intertwine}
\begin{aligned}
P \iota(r^{-1}) \circ \rho^-_{q^{-1}z,q^{-1}r}(u) &= (\pi_z \ot \rho^-_{z,r})(\Del(u)) \circ P \iota(r^{-1}) \\
\rho^-_{q z, q r}(u) \circ \tau(r^{-1}) P &= \tau(r^{-1}) P \circ  (\pi_z \ot \rho^-_{z,r} )(\Del(u))
\end{aligned}
\Bigg\} & \qu \text{for all } u \in U_q(\mfb^-).
\end{align}
Solving \eqrefs{iotatauplus:intertwine}{iotataumin:intertwine} for $u=k_0$ and $u=k_1$ we obtain, using \eqref{oscq:decomposition}, that
\[
\iota(r) = \begin{pmatrix} a^\dag f_0(D) \\ f_1(D) \end{pmatrix} \qq \tau(r) = \begin{pmatrix} g_0(D), \, a^\dag g_1(D) \end{pmatrix} 
\]
respectively, where $f_0, f_1, g_0, g_1 \in \mc F$ are arbitrary.
Now solving \eqrefs{iotatauplus:intertwine}{iotataumin:intertwine} for $u \in \{ e_0,e_1\}$ and using \eqref{rhoplus:def}, as well as for $u \in \{ f_0,f_1 \}$ and using \eqref{rhominus:def}, we deduce \eqref{iotatau:formula}.

To prove the short exact sequence statements, we straightforwardly verify that $\iota(r)$ is injective, $\tau(r)$ is surjective and the image of $\iota(r)$ equals the kernel of $\tau(r)$, respectively.
\end{proof}

The \emph{(bulk) fusion relations} are the following identities relating $\iota(r)$ and $\tau(r)$ to $R(z)$ and $L(z)$ can be obtained directly from the formulas \eqref{iotatau:formula} and \eqrefs{R:def}{L:def}:
\begin{align}
\label{L:iota}
L_{13}(z,r) R_{23}(z) (\iota(r) \ot \Id) &= (1 - z^2) (\iota(r) \ot \Id) L(q z, q r) & \in \Hom(W \ot V,W \ot V \ot V), \\
\label{L:tau}
(\tau(r) \ot \Id) L_{13}(z,r) R_{23}(z) &= q (1 - q^2 z^2) L(q^{-1} z, q^{-1} r) (\tau(r) \ot \Id) \hspace{-12mm} \\
\nonumber && \in \Hom(W \ot V \ot V,W \ot V).
\end{align}
Alternatively, if one applies $\rho^+_{z,r} \ot \pi_z \ot \pi_1$ to \eqref{uniR:2} and uses \eqref{iotatauplus:intertwine} or $\rho^-_{z^{-1},r^{-1}} \ot \pi_{z^{-1}} \ot \pi_1$ to $(\Del^{\rm op} \ot \id)(\mc R_{21}) = \mc R_{31} \mc R_{32}$ (itself a direct consequence of \eqref{uniR:3}) then one obtains these identities up to the scalar factor, which can then be found by evaluating on a particular nonzero vector.\\

We also have \emph{boundary fusion relations} similar to the finite-dimensional relation given in \cite[Eqn.~(4.7)]{RSV16}.
Boundary fusion was first discussed in the papers \cite{KS92,MN92} without explicitly using short exact sequences.

\begin{lemma} \label{KW:iotatau:lem}
We have the following identities in $\Hom(W,W \ot V)$ and $\Hom(W \ot V,W)$, respectively:
\begin{align} 
\label{KW:iota} K^W_1(z,r) L(z^2,r) K^V_2(z) \iota(r) &= (1 - z^4) (\xi - q^2 z^2) \iota(r) K^W(q z,q r), \\
\label{KW:tau} \tau(r) K^W_1(z,r) L(z^2,r) K^V_2(z) &= r (\xi z^2  - 1) K^W(q^{-1} z,q^{-1} r) \tau(r)
\end{align}
and the following identities in $\Hom(W,W \ot V)$ and $\Hom(W \ot V,W)$, respectively:
\begin{align} 
\label{KtW:iota} \wt K^V_2(z) \wt L(z^2,r) \wt K^W_1(z,r) \iota(r) &= \frac{\tilde \xi - q^2 z^2}{1 - q^2 z^4} \, \iota(r) \wt K^W(q z,q r), \\
\label{KtW:tau} \tau(r) \wt K^V_2(z) \wt L(z^2,r) \wt K^W_1(z,r) &= r^{-1} \frac{1 - q^4 z^4}{1 - q^2 z^4} (\tilde \xi z^2 - 1) \wt K^W(q^{-1} z,q^{-1} r) \tau(r).
\end{align}
\end{lemma}

\begin{proof}
A direct computation gives 
\begin{align*}
L(z^2,r) K^V_2(z) \iota(r) &= \begin{pmatrix} 1 & -q^{-1} z^2 a^\dag \\ - q z^2 a & 1-q^{2(D+1)} z^4 \end{pmatrix} \begin{pmatrix} q^D r & 0 \\ 0 & q^{-D} \end{pmatrix} \begin{pmatrix} \xi z^2 - 1 & 0 \\ 0 & \xi - z^2 \end{pmatrix} \begin{pmatrix} q^{-D} a^\dag \\ q^{D+1} r \end{pmatrix} \\
&= r \begin{pmatrix} 1 & -q^{-1} z^2 a^\dag \\ - q z^2 a & 1-q^{2(D+1)} z^4 \end{pmatrix} \begin{pmatrix} (\xi z^2 - 1) a^\dag \\ q (\xi - z^2) \end{pmatrix} \\
&= r (z^4 - 1) \begin{pmatrix} a^\dag \\  q (q^{2(D+1)}z^2-\xi)\end{pmatrix} 
\end{align*}
so that
\begin{align*}
K^W_1(z,r) L(z^2,r) K^V_2(z) \iota(r) &= q^{-D^2} r^{1-D} (z^4 - 1)  \begin{pmatrix} a^\dag  \\  q \end{pmatrix} (-\xi)^{D+1} (q^2 z^2/\xi)_{D+1} \\
&= (z^4 - 1) (q^2 z^2 - \xi) \begin{pmatrix} q^{-D} a^\dag \\  q^{D+1} r  \end{pmatrix} q^{-D^2} (qr)^{-D}  (-\xi)^D (q^4 z^2/\xi)_D \\
&= (z^4 - 1) (q^2 z^2 - \xi) \iota(r) K^W(q z,q r),
\end{align*}
as required.
We also have
\begin{align*}
\tau(r) K^W_1(z,r) &= \begin{pmatrix} q^D & -q^{-D} r^{-1} a^\dag \end{pmatrix} q^{-D^2} r^{-D} (-\xi)^D (q^2 z^2/\xi)_D \\
&= (-\xi)^{D-1} (q^2 z^2/\xi)_{D-1} \begin{pmatrix} q^{2D} z^2 - \xi & -q^{-1} a^\dag \end{pmatrix} \\
&= (z^2 - \xi)^{-1} K^W(q^{-1} z) \begin{pmatrix} q^{2D} z^2 - \xi & -q^{-1} a^\dag \end{pmatrix} 
\end{align*}
so that
\begin{align*}
& \tau(r) K^W_1(z,r)  L(z^2,r) K^V_2(z) = \\
&= (z^2 - \xi)^{-1} K^W(q^{-1} z,q^{-1} r) \begin{pmatrix} z^2 - \xi  &  q^{-1} (\xi z^2 - 1) a^\dag   \end{pmatrix} \begin{pmatrix} (\xi z^2 - 1) q^D r & 0 \\ 0 & (\xi - z^2) q^{-D} \end{pmatrix} \\
&= r (\xi z^2 - 1) K^W(q^{-1} z, q^{-1} r) \begin{pmatrix} q^D &  -q^{-D} r^{-1} a^\dag  \end{pmatrix} 
\end{align*}
as required.
This gives  \eqrefs{KW:iota}{KW:tau}.
Now \eqrefs{KtW:iota}{KtW:tau} straightforwardly follow by applying \eqref{L:cross.unit.} and \eqref{Kt:fromK}.
\end{proof}


\section{Transfer matrices and the Q-operator}

In this section we define, and prove properties of, the main objects of this paper: the transfer matrix and the Q-operator for the open XXZ chain with diagonal boundaries.
Initially we will define the Q-operator as a formal power series and in Section \ref{sec:Betheeqns} we will show it converges for suitable parameters and hence determines a well-defined linear map.

\subsection{Double-row operators}

Fix $\bm t = (t_1,\ldots,t_N) \in (\C^\times)^N$.
We will construct operators acting on $V^{\ot N}$ by first considering operators called \emph{quantum monodromy matrices} acting on a larger tensor product $V \ot V^{\ot N}$ or $W \ot V^{\ot N}$.
The additional factor in the tensor product is called the \emph{auxiliary space} - in the construction of the transfer matrix this space will be traced out.
The quantum monodromy matrices are compositions of operators each acting on the auxiliary space and at most one other space, as follows:
\begin{align}
\label{MV:def}
\mc M^V_a(z) &:= R_{1a}(t_1 z) \cdots R_{Na}(t_N z) K^V_a(z) R_{aN}(\tfrac{z}{t_N}) \cdots R_{a1}(\tfrac{z}{t_1}) && \in \End(V \ot V^{\ot N}) \hspace{-8pt} \\
\label{MW:def}
\hspace{-4pt} \mc M^W_a(z,r) &:= L^-_{1a}(t_1 z,r) \cdots L^-_{Na}(t_N z,r) K^W_a(z,r) L^+_{aN}(\tfrac{z}{t_N},r) \cdots L^+_{a1}(\tfrac{z}{t_1},r) \hspace{-4pt} && \in \End(W \ot V^{\ot N}).
\hspace{-6pt}
\end{align}
Here we have labelled the auxiliary space by $a$ (in general we will use lowercase roman letters to label auxiliary spaces).
Note that due to the explicit expressions \eqrefs{R:def}{L:def} we have
\begin{align*}
\mc M^V_a(z) &= R_{a1}(t_1 z) \cdots R_{aN}(t_N z) K^V_a(z) R_{aN}(\tfrac{z}{t_N}) \cdots R_{a1}(\tfrac{z}{t_1}) \\
\mc M^W_a(z,r) &= L_{a1}(t_1 z,r) \cdots L_{aN}(t_N z,r) K^W_a(z,r) L_{aN}(\tfrac{z}{t_N},r) \cdots L_{a1}(\tfrac{z}{t_1},r).
\end{align*}
These quantum monodromy matrices are particular to open chains; the operator $\mc M^V(z)$ was first considered by Sklyanin \cite{Sk88}.
In particular he showed it satisfies a version of the reflection equation in $\End(V \ot V \ot V^{\ot N})$; there is a similar identity in $\End(W \ot V \ot V^{\ot N})$ involving both $\mc M^V(z)$ and $\mc M^W(z,r)$, see \eqref{MW:RE} below.
In such identities there are two auxiliary spaces which are labelled $a$ and $b$, respectively.

\begin{lemma} \label{lem:MW:RE}
In $\End(W \ot V \ot V^{\ot N})$ we have the identities
\begin{align}
\label{MRM1} \mc M^W_a(y,r) L_{ab}(yz,r) \mc M^V_{b}(z) &= \big( L_{a1}(t_1 y,r) R_{b1}(t_1 z) \big) \cdots \big( L_{aN}(t_N y,r) R_{bN}(t_N z) \big) \cdot \\
\nonumber & \qu \cdot K^W_a(y,r) L_{ab}(yz,r) K^V_{b}(z) \cdot \\
\nonumber & \qu \cdot \big( L_{aN}(\tfrac{y}{t_N},r) R_{bN}(\tfrac{z}{t_N}) \big) \cdots \big( L_{a1}(\tfrac{y}{t_1},r) R_{b1}(\tfrac{z}{t_1}) \big) \hspace{-5mm} \\
\label{MRM2} \mc M^V_{b}(z) L_{ab}(yz,r) \mc M^{W}_a(y,r) &= \big( R_{b1}(t_1 z) L_{a1}(t_1 y,r) \big) \cdots \big( R_{bN}(t_N z) L_{aN}(t_N y,r) \big) \cdot \\
\nonumber & \qu \cdot K^V_{b}(z) L_{ab}(yz,r) K^W_a(y,r) \cdot \\
\nonumber & \qu \cdot \big( R_{bN}(\tfrac{z}{t_N}) L_{aN}(\tfrac{y}{t_N},r) \big) \cdots \big( R_{b1}(\tfrac{z}{t_1}) L_{a1}(\tfrac{y}{t_1},r) \big) \hspace{-5mm} 
\end{align}
and
\eq{
\label{MW:RE} \hspace{-3mm} L_{ab}(\tfrac{y}{z},r) \mc M^W_a(y,r) L_{ab}(y z,r) \mc M^V_{b}(z) = \mc M^V_{b}(z) L_{ab}(y z,r) \mc M^W_a(y,r) L_{ab}(\tfrac{y}{z},r).
}
\end{lemma}

\begin{proof}
We obtain \eqrefs{MRM1}{MRM2} as direct consequences of the definition \eqref{MW:def} by applying \eqref{Lplus:YBE} and \eqref{Lmin:YBE} $N$ times.
To obtain \eqref{MW:RE}, we use that \eqrefs{MRM1}{MRM2} provide factorized expressions for its left- and right-hand sides, respectively.
More precisely, \eqref{MW:RE} follows from the remark that left-multiplying \eqref{MRM1} by $L_{ab}(\tfrac{y}{z},r)^{-1}$ and right-multiplying it by $L_{ab}(\tfrac{y}{z},r)$, applying \eqref{Lmin:YBE} $N$ times, \eqref{KW:RE} once and \eqref{Lplus:YBE} $N$ times yields \eqref{MRM2}.
\end{proof}

The following lemma is vital to the proof of the main theorem of this paper, Theorem \ref{thm:QT:rel}.

\begin{lemma} \label{lem:iotatau:row}
We have the following identities in $\End(W \ot V^{\ot N},W \ot V \ot V^{\ot N})$ and $\End(W \ot V \ot V^{\ot N},W \ot V^{\ot N})$, respectively:
\begin{align}
\label{iota:row} & \wt K^V_{b}(z) \wt L_{ab}(z^2,r) \wt K^W_a(z,r) \mc M^W_a(z,r) L_{ab}(z^2,r) \mc M^V_{b}(z)  (\iota(r) \ot \Id_{V^{\ot N}})= \\
\nonumber & \qq = \frac{p_+(z)}{1-q^2 z^4} (\iota(r) \ot \Id_{V^{\ot N}}) \wt K^W_a(qz,qr) \mc M^W_a(qz,qr), \\
\label{tau:row} & (\tau(r) \ot \Id_{V^{\ot N}}) \wt K^V_{b}(z) \wt L_{ab}(z^2,r) \wt K^W_a(z,r) \mc M^W_a(z,r) L_{ab}(z^2,r) \mc M^V_{b}(z)  = \\
\nonumber & \qq = \frac{p_-(z)}{1-q^2 z^4} \wt K^W_a(q^{-1} z,q^{-1} r) \mc M^W_a(q^{-1} z , q^{-1} r) (\tau(r) \ot \Id_{V^{\ot N}})
\end{align}
where we have introduced the coefficient polynomials
\eq{ \label{p:def}
\begin{aligned}
p_+(z) &:= (1 - z^4) (\tilde \xi  - q^2 z^2) (\xi - q^2 z^2)  \prod_{n=1}^N (1 - t_n^2 z^2)(1 - t_n^{-2} z^2) \\
p_-(z) &:= q^{2N} (1 - q^4 z^4) (1 - \tilde \xi z^2) (1 - \xi z^2) \prod_{n=1}^N (1 - q^2 t_n^2 z^2)(1 - q^2 t_n^{-2} z^2).
\end{aligned}
}
\end{lemma}

\begin{proof}
Take \eqref{MRM1} in the special case $y=z$.
Right-multiplying it by $\iota(r) \ot \Id_{V^{\ot N}}$, using \eqref{L:iota} $N$ times, \eqref{KW:iota} once and \eqref{L:iota} another $N$ times we arrive at 
\begin{align}
\label{iota:row:pre} & \mc M^W_a(z,r) L_{ab}(z^2,r) \mc M^V_{b}(z) (\iota(r) \ot \Id_{V^{\ot N}}) = \\
\nonumber & \qu = (1 - z^4) (\xi  - q^2 z^2) \left( \prod_{n=1}^N (1 - t_n^2 z^2)(1 - t_n^{-2} z^2) \right) (\iota(r) \ot \Id_{V^{\ot N}}) \mc M^W_a(qz,qr).
\end{align}
On the other hand, left-multiplying \eqref{MRM1} with $y=z$ by $\tau(r) \ot \Id_{V^{\ot N}}$, using \eqref{L:tau} $N$ times, \eqref{KW:tau} once and \eqref{L:tau} another $N$ times, we arrive at 
\begin{align}
\label{tau:row:pre} & (\tau(r) \ot \Id_{V^{\ot N}}) \mc M^W_a(z,r) L_{ab}(z^2,r) \mc M^V_{b}(z) = \\
\nonumber & \qu = q^{2N} (\xi z^2 - 1) \left( \prod_{n=1}^N (1 - q^2 t_n^2 z^2)(1 - q^2 t_n^{-2} z^2) \right) \mc M^W_a(q^{-1} z , q^{-1} r) (\tau(r) \ot \Id_{V^{\ot N}}).
\end{align}
Combining \eqrefs{iota:row:pre}{tau:row:pre} with \eqrefs{KtW:iota}{KtW:tau} we obtain \eqrefs{iota:row}{tau:row}.
\end{proof}

We record some more properties of $\mc M^W(z,r)$ which we will need later on.
As an immediate consequence of \eqrefs{R:def}{L:def}, \eqref{KV:def} and \eqref{KW:def} we obtain
\[
\left[ \Big( \begin{smallmatrix} 1 & 0 \\ 0 & y \end{smallmatrix} \Big)^{\ot 2}, R(z) \right] = \left[ \Big( \begin{smallmatrix} y^D & 0 \\ 0 & y^{D+1} \end{smallmatrix} \Big), L(z,r) \right] = \left[ \Big( \begin{smallmatrix} 1 & 0 \\ 0 & y \end{smallmatrix} \Big), K^V(z) \right] = \left[ y^D, K^W(z,r) \right] = 0
\]
which implies that for all $r \in \C^\times$ and $y,z \in \C$ we have
\eq{ \label{M:totalspin}
\left[ \Big( \begin{smallmatrix} 1 & 0 \\ 0 & y \end{smallmatrix} \Big) \ot \Big( \begin{smallmatrix} 1 & 0 \\ 0 & y \end{smallmatrix} \Big)^{\ot N} , \mc M^V(z) \right] = \left[ y^D \ot \Big( \begin{smallmatrix} 1 & 0 \\ 0 & y \end{smallmatrix} \Big)^{\ot N} , \mc M^W(z,r) \right] = 0.
}
From \eqref{L:def} we obtain
\[
L(z,r) = L(z,1) \Big( \begin{smallmatrix} r & 0 \\ 0 & 1 \end{smallmatrix} \Big) = \Big( \begin{smallmatrix} r & 0 \\ 0 & 1 \end{smallmatrix} \Big) r^{-D} L(z,1) r^D
\]
and from \eqref{KW:def} and \eqref{KtW:def} we have
\[
K^W(z,r) = r^{-D} K^W(z,1), \qq \wt K^W(z,r) =  \wt K^W(z,1) r^D
\]
and hence we have the following factorization for all $r \in \C^\times$, $z \in \C$:
\eq{ \label{MW:r-dependence}
\wt K^W(z,r) \mc M^W(z,r) =  \left( \Id \ot \Big( \begin{smallmatrix} r & 0 \\ 0 & 1 \end{smallmatrix} \Big)^{\ot N} \right) \wt K^W(z,1) \mc M^W(z,1) \left( \Id \ot \Big( \begin{smallmatrix} r & 0 \\ 0 & 1 \end{smallmatrix} \Big)^{\ot N} \right) .
}

\subsection{Definition and basic properties of the transfer matrices and the Q-operator}

The following three elements of $\End(V^{\ot N})$ are central to the notion of quantum integrability for the open XXZ Heisenberg spin chain.

\begin{defn}
Let $\bm t \in (\C^\times)^N$.
The \emph{transfer matrices} (with respect to auxiliary spaces $V$ and $W$, respectively) are the following linear maps on $V^{\ot N}$:
\eq{ \label{T:def} 
\mc T^V(z) := \Tr_a  \wt K^V_a(z) \mc M^V_a(z), \qq \qq \mc T^W(z) := \Tr_a \wt K^W_a(z,1) \mc M^W_a(z,1)  
}
where the label of the space over which we trace is indicated by a subscript.
Moreover the \emph{Q-operator} is the linear map
\begin{flalign}
\label{Q:def} 
\mc Q(z) := \Big( \begin{smallmatrix} z^2 & 0 \\ 0 & 1 \end{smallmatrix} \Big)^{\ot N} \mc T^W(z) \in \End(V^{\ot N}).
\end{flalign}
For now we treat these operators as formal power series in $z$.
\hfill \defnend
\end{defn}

\begin{rmk} \mbox{}
\begin{enumerate}
\item The operator $\mc T^W(z)$ corresponds, up to a scalar factor and a renumbering of the $t_n$, to the object defined in \cite[Eq. (G.24), $a=1$]{BTs18}.
\item Note that we have set $r=1$ in the definition of $\mc T^W(z)$ and hence $\mc Q(z)$.
The relations \eqrefs{iota:row}{tau:row}, where $r$ varies, will be used to derive Baxter's functional relation, see \eqref{QT:rel}, and to compensate for this it is necessary, as will become apparent, to introduce the extra diagonal matrix in the definition of $\mc Q(z)$.
\hfill \rmkend
\end{enumerate}
\end{rmk}

As a consequence of the properties
\eq{ 
\begin{gathered}
R(-z) = \bigg( \Big( \begin{smallmatrix} 1 & 0 \\ 0 & -1 \end{smallmatrix} \Big) \ot \Id_V \bigg) R(z)  \bigg( \Big( \begin{smallmatrix} 1 & 0 \\ 0 & -1 \end{smallmatrix} \Big) \ot \Id_V \bigg)  =  \bigg(\Id_V \ot \Big( \begin{smallmatrix} 1 & 0 \\ 0 & -1 \end{smallmatrix} \Big) \bigg) R(z) \bigg(\Id_V \ot \Big( \begin{smallmatrix} 1 & 0 \\ 0 & -1 \end{smallmatrix} \Big) \bigg), \\
L(-z,r) = \big( (-1)^D \ot \Id_V \big) L(z,r) \big( (-1)^D \ot \Id_V \big) = \bigg(\Id_W \ot \Big( \begin{smallmatrix} 1 & 0 \\ 0 & -1 \end{smallmatrix} \Big) \bigg) L(z,r) \bigg(\Id_W \ot \Big( \begin{smallmatrix} 1 & 0 \\ 0 & -1 \end{smallmatrix} \Big) \bigg), \\
K^V(-z) = K^V(z), \qq \wt K^V(-z) = \wt K^V(z) \\
K^W(-z,r) = K^W(z,r), \qq \wt K^W(-z,r) = \wt K^W(z,r)
\end{gathered}
}
and cyclicity of the trace we obtain that the matrix-valued formal power series $\mc Q(z)$, $\mc T^W(z)$, $\mc T^V(z)$ are invariant under $z \mapsto -z$ and hence we deduce that they are series in $z^2$.

Similarly, as a direct consequence of \eqref{M:totalspin}, the properties
\[
\left[ \Big( \begin{smallmatrix} y & 0 \\ 0 & 1 \end{smallmatrix} \Big), \wt K^V(z) \right] = \left[ y^D, \wt K^W(z,r) \right] = \left[ \Big( \begin{smallmatrix} y & 0 \\ 0 & 1 \end{smallmatrix} \Big),\Big( \begin{smallmatrix} y & 0 \\ 0 & 1 \end{smallmatrix} \Big) \right] = 0 
\]
and cyclicity of the trace we obtain
\eq{ \label{QT:totalspin}
\left[ \mc T^V(z),\Big( \begin{smallmatrix} y & 0 \\ 0 & 1 \end{smallmatrix} \Big)^{\ot N} \right] = 
\left[ \mc T^W(z), \Big( \begin{smallmatrix} y & 0 \\ 0 & 1 \end{smallmatrix} \Big)^{\ot N} \right] =
\left[ \mc Q(z),\Big( \begin{smallmatrix} y & 0 \\ 0 & 1 \end{smallmatrix} \Big)^{\ot N} \right] = 0
}
understood as an equation of matrix-valued formal power series in $y$ and $z$.
As a consequence, 
\[
[ \mc Q(y),\mc T^V(z) ] = \Big( \begin{smallmatrix} y^2 & 0 \\ 0 & 1 \end{smallmatrix} \Big)^{\ot N} \Big[ \mc T^W(y), \mc T^V(z) \Big], \qq
[ \mc Q(y),\mc Q(z) ] = \Big( \begin{smallmatrix} y^2 z^2 & 0 \\ 0 & 1 \end{smallmatrix} \Big)^{\ot N} \Big[ \mc T^W(y), \mc T^W(z) \Big].
\]

Note that $\mc T^V(z)$ is a well-defined linear operator on $V^{\ot N}$ whose entries depend polynomially on $z$ (they are finite sums of certain entries of products of matrices whose entries depend polynomially on $z$).
We summarize the above discussion in the following Lemma.

\begin{lemma} \label{lem:TQ:basicproperties}
The matrix $\mc T^V(z)$ depends polynomially on $z^2$.
The matrices $\mc T^W(z)$ and $\mc Q(z)$ are formal power series in $z^2$.
All these operators commute with $\Big( \begin{smallmatrix} y & 0 \\ 0 & 1 \end{smallmatrix} \Big)^{\ot N}$ for all $y \in \C$.
\end{lemma}

The polynomiality of $\mc T^W(z)$ and $\mc Q(z)$ is not obvious.
For now we show that for suitable values of $q$, $\xi$ and $\tilde \xi$ and all but finitely many $z$ the matrix entries of $\mc T^W(z)$ and $\mc Q(z)$, are well-defined (i.e.~the series associated with the trace converges).
In order to do this, for $q,x$ inside the unit disk, $a,b \in \C$ and $c \in \C \backslash q^{2 \Z_{\le 0}}$ we consider the basic hypergeometric function (see e.g.~\cite[1.2]{GR04})
\[
\phi(a,b,c;x) := {}_2 \phi_1\left( {a, b \atop c} ; q^2 , x \right) = \sum_{j \in \Z_{\ge 0}} \frac{(a)_j (b)_j}{(q^2)_j (c)_j} x^j
\]
which converges absolutely, as a consequence of the ratio test.
Consider the open set
\eq{
\label{S:def} S^{(N)} := \{ (a,b,c) \in (\C^\times)^3 \, | \, |a|<1, \, |b c| < |a|^{2N} \}.
}
For $(q,\xi,\tilde \xi) \in S^{(N)}$, the set $\{ \pm \xi^{1/2}, \pm q \xi^{1/2}, \ldots, \pm q^{N-1} \xi^{1/2}, \pm q^{-1} \tilde \xi^{-1/2}, \pm q^{-2} \tilde \xi^{-1/2}, \ldots \} $ is discrete and we consider its complement
\eq{
\label{U:def} U^{(N)} := \C \backslash \{ \pm \xi^{1/2}, \pm q \xi^{1/2}, \ldots, \pm q^{N-1} \xi^{1/2}, \pm q^{-1} \tilde \xi^{-1/2}, \pm q^{-2} \tilde \xi^{-1/2}, \ldots \} .
}

\begin{thrm} \label{thm:Q:convergent}
The matrix entries of $\mc T^W(z)$ and $\mc Q(z)$ converge provided $(q, \xi , \tilde \xi) \in S^{(N)}$, $\bm t \in (\C^\times)^N$ and $z \in U^{(N)}$.
The resulting expressions depend holomorphically on each of $q,\xi,\tilde \xi$, Laurent polynomially on each $t_n$ and holomorphically on $z^2$; at each $z \in U^{(N)}$ the expressions have simple poles or removable singularities.
\end{thrm}

\begin{proof}
By \eqref{Q:def} it suffices to prove the result for $\mc T^W(z)$.

\emph{Claim:} Any matrix entry of $\wt K^W_a(z,1) \mc M^W_a(z,1)$ is of the form
\eq{ \label{Q:convergent:1}
\frac{1}{1-q^2 \tilde \xi z^2} \sum_{L,M \in \Z \atop |L| \le N, -N \le M \le 2N} c_{L,M}(z^2;\bm t) \frac{(q^{2(L+1)}z^2/\xi)_D}{(q^4 \tilde \xi z^2)_D} (q^{2M} \xi \tilde \xi)^D + \text{terms in } \Ker\Big(\Tr_W\Big)
}
for some $c_{L,M}(z^2;\bm t) \in \C$ which depends holomorphically on each of $q,\xi,\tilde \xi$, Laurent polynomially on each $t_n$ and meromorphically on $z^2$, with all poles, if any, simple and at $z^2 = q^{-2} \tilde \xi^{-1}$ and $z^2 = q^{2i} \xi$ for $i \in \{ 0,1,\ldots,N-1 \}$.
Note that owing to the ratio of q-Pochhammer symbols there may be simple poles at $z^2 = q^{-2k} \tilde \xi^{-1}$ for $k \in \Z_{\ge 2}$.
The claim implies that an arbitrary matrix entry of $\mc T^W(z)$ is of the form
\[ 
\sum_{L,M \in \Z \atop |L| \le N, -N \le M \le 2N} c_{L,M}(z^2;\bm t) \phi(q^{2L}z^2/\xi,q^2,q^4 \tilde \xi z^2;q^{2M}\xi \tilde \xi)
\]
so that the condition on the norm of $\xi \tilde \xi$ now guarantees convergence and the statement of the theorem follows.

To prove the claim, we analyse the occurrence of powers of $q^D$, $a$ and $a^\dag$ in matrix entries of $\mc M(z,1)$.
We start with a combinatorial analysis of how the matrix entries of the L-operators appear in these entries.
For $\ga,\del \in \{0,1\}$ we define matrix entries $L(z)_\ga^\del \in \End(W)$ as follows:
\[
L(z,1) (w \ot v^\ga) = \sum_{\del \in \{0,1\}} L(z)^\ga_\del(w) \ot v^\del.
\]
for all $w \in W$.
By \eqref{L:def} we have
\eq{ \label{L:entries}
\begin{array}{rlrl}
L(z)_0^0 &= q^D, & L(z)_0^1 &= -z q^{-D} a^\dag, \\
L(z)_1^0 &= -z a q^{D+1}, & L(z)_1^1 &= q^{-D} (1- q^{2(D+1)} z^2).
\end{array}
}
Matrix entries of $\mc M(z,1)$ are given by
\eq{ \label{M:entry}
\mc M(z)^{\veps_1, \ldots, \veps_N}_{\ga_1, \ldots, \ga_N} = \sum_{\del_1,\ldots,\del_N \in \{0,1\}} L(t_1 z)_{\ga_1}^{\del_1} \cdots L(t_N z)_{\ga_N}^{\del_N} K^W(z,1) L(z/t_N)_{\del_N}^{\veps_N} \cdots L(z/t_1)_{\del_1}^{\veps_1}
}
Consider one of the terms in this sum and let $n \in \{1,2,\ldots,N\}$ be arbitrary.
We focus our attention on the factors $L(t_n z)_{\ga_n}^{\del_n}$ and $L(z/t_n)_{\del_n}^{\veps_n}$ in the summand
\eq{ \label{summand}
L(t_1 z)_{\ga_1}^{\del_1} \cdots L(t_n z)_{\ga_n}^{\del_n} \cdots L(t_N z)_{\ga_N}^{\del_N} K^W(z,1) L(z/t_N)_{\del_N}^{\veps_N} \cdots L(z/t_n)_{\del_n}^{\veps_n} \cdots L(z/t_1)_{\del_1}^{\veps_1}.
}
There are 8 possibilities for each pair $(L(t_n z)_{\ga_n}^{\del_n},L(z/t_n)_{\del_n}^{\veps_n})$ and these are enumerated by triples $(\ga_n,\del_n,\veps_n) \in \{ 0,1\}^3$.
Now letting $n$ runs through $\{1,\ldots,N\}$, we denote by $j(\ga,\del,\veps) \in \Z_{\ge 0}$ the number of pairs $( L(t_n z)_{\ga}^{\del},  L(z/t_n)_{\del}^{\veps} )$ occurring in the summand \eqref{summand}; in other words $j(\ga,\del,\veps)$ is the cardinality of the set $\{ n \, | \, (\ga_n,\del_n,\veps_n) = (\ga,\del,\veps) \}$.
For instance, for $N=2$ the summand $L(t_1 z)_0^0 L(t_2 z)_0^0 K^W(z,1) L(z/t_2)_0^1 L(z/t_1)_0^0$ has $j(0,0,0)=j(0,0,1)=1$ with the other $j(\ga,\del,\veps)$ equal to zero.
Note that the sum of all or some of the $j(\ga,\del,\veps)$ is bounded above by $N$.
Therefore each matrix entry of $\mc M(z,1)$ is a sum of terms of the form
\[
A^-(z;\bm t) K^W(z,1) A^+(z;\bm t)
\]
where $A^-(z;\bm t)$ is a product, in some order, of various factors $q^D$, $-t_n z q^{-D} a^\dag$, $-t_n z a q^{D+1}$ and $q^{-D}(1 - q^{2(D+1)} t_n^2 z^2)$.
By the definition of $j(\ga,\del,\veps)$ we know precisely how many there are of each: 
$j(0,0,0) + j(0,0,1)$ factors $q^D$, $j(0,1,0)+j(0,1,1)$ factors $-t_n z q^{-D} a^\dag$, $j(1,0,0)+j(1,0,1)$ factors $-t_n z a q^{D+1}$ and $j(1,1,0)+j(1,1,1)$ factors $q^{-D}(1 - q^{2(D+1)} t_n^2 z^2)$.
Similarly, $A^+(z;\bm t)$ is a product, in some order, of $j(0,0,0)+ j(1,0,0)$ factors $q^D$, $j(0,0,1)+j(1,0,1)$ factors $-t_n^{-1} z q^{-D} a^\dag$, $j(0,1,0)+j(1,1,0)$ factors $-t_n^{-1} z a q^{D+1}$ and $j(0,1,1)+j(1,1,1)$ factors $q^{-D}(1 - q^{2(D+1)} t_n^{-2} z^2)$.
Since $\Tr_W a^m f(D) = \Tr_W (a^\dag)^m f(D) = 0$ for any $m \in \Z_{>0}$ and any $f \in \mc F$, only terms that satisfy
\[
j(0,0,1) + j(0,1,1) = j(1,0,0) + j(1,1,0)
\]
will contribute to the trace over $W$.
As a consequence, the power of $z$ coming from the factors proportional to $a$ and $a^\dag$ is even so that the $z$-dependence of arbitrary matrix entries of $\mc M^W(z,1)$ that contribute to the trace over $W$ is only through $z^2$.
Next we use the relations \eqref{oscq:relations} to move $K^W(z,1) = q^{-D^2}  (-\xi)^D (q^2 z^2/\xi)_D $ to the left, all overall powers of $q^D$ in $A^-(z;\bm t)$ to the left and all overall powers of $q^D$ in $A^+(z;\bm t) $ to the right.
We do not move factors of the form $1 - q^{2(D+1)} t_n^{\pm 2} z^2$.
It yields that any matrix entry of $\mc M^W(z,1)$ is a linear combination, with coefficients depending polynomially on $z^2$, Laurent polynomially on $q$ and Laurent polynomially on each $t_n$, of expressions of the form
\[
(q^2 z^2/\xi)_{D+j_-} q^{-(D+j_-)^2} q^{(j_+ + j_-)D} (-\xi)^D \wt A(z^2;\bm t) q^{j_{\rm alt} D} 
\]
plus terms in the kernel of $\Tr_W$.
Here we have introduced the shorthand notations
\begin{align*}
j_- &:= -j(0,1,0) - j(0,1,1) + j(1,0,0) + j(1,0,1), \\
j_+ &:= j(0,0,0) + j(0,0,1) - j(1,1,0) - j(1,1,1), \\
j_{\rm alt} &:= j(0,0,0) - j(0,0,1) + j(0,1,0) - j(0,1,1) + j(1,0,0) - j(1,0,1) + j(1,1,0) - j(1,1,1)
\end{align*}
and $\wt A(z^2;\bm t)$ a product, in some order, of 
\[
j(0,0,1)+j(0,1,0)+j(0,1,1)+j(1,0,1) = j(0,1,0)+j(1,0,0)+j(1,0,1)+j(1,1,0)
\]
factors $a$, the same number of factors $a^\dag$, $j(1,1,0)+j(1,1,1)$ factors of the form $1-q^{2(D+1)} t_n^2 z^2$ and $j(0,1,1)+j(1,1,1)$ factors of the form $1-q^{2(D+1)} t_n^{-2} z^2$.

Using \eqref{oscq:relations} once again, we deduce that, for any $m \in \Z_{\ge 0}$, $(a^\dag)^m a^m$ and $a^m (a^\dag)^m$ are linear combinations of $1, q^{2D}, q^{4D} \ldots, q^{2mD}$.
Since the total number of pairs $\{ a,a^\dag\}$ plus the total number of factors of the form $1-q^{2(D+1)} t_n^{\pm 2} z^2$ is bounded above by $N$, it follows that 
\[
\wt A(z^2;\bm t) = \sum_{M'=0}^N c_{M'}(z^2;\bm t) q^{2M'D}
\]
where $c_{M'}(z^2;\bm t) \in \C$ depends polynomially on $q^2$ and $z^2$ and Laurent polynomially on each $t_n^2$, so that, up to terms in $\Ker(\Tr_W)$, any matrix entry of $\mc M^W(z,1)$ is a linear combination, with coefficients depending polynomially on $z^2$, Laurent polynomially on $q$ and Laurent polynomially on each $t_n$, of expressions of the form
\[
(q^2 z^2/\xi)_{D+j_-} q^{-(D+j_-)^2} q^{(j_++j_-)D} (-\xi)^D \sum_{M'=0}^N c_{M'}(z^2;\bm t) q^{(2M'+j_{\rm alt})D} .
\]
We denote
\[
j_0 := \frac{1}{2}(j_{\rm alt} + j_+ - j_-) = j(0,0,0) + j(0,1,0) - j(1,0,1) - j(1,1,1).
\]
Re-arranging the powers of $q$ we obtain that, modulo terms in $\Ker(\Tr_W)$, any matrix entry of $\mc M^W(z,1)$ is a linear combination, with coefficients depending polynomially on $z^2$, Laurent polynomially on $q$ and Laurent polynomially on each $t_n$, of expressions of the form
\[
(q^2 z^2/\xi)_{D+j_-} q^{-D^2} \sum_{M'=0}^N c_{M'}(z^2;\bm t) q^{2(M' + j_0)D} (-\xi)^D.
\]

Recall from \eqref{KtW:def} that $\wt K^W(z,1) = (1-q^2 \tilde \xi z^2)^{-1} q^{D^2} (-\tilde \xi)^D (q^4 \tilde \xi z^2)_D^{-1}$.
Also note the basic property $(x)_{j+k} = (x)_j (q^{2j} x)_k$ for all $j,k \in \Z$, see \cite[(1.2.33)]{GR04}.
Hence we deduce that an arbitrary matrix entry of $\wt K^W(z,1) \mc M^W(z,1)$ is a linear combination of expressions of the form
\[
 (q^2 z^2/\xi)_{j_-} \frac{(q^{2(j_-+1)} z^2/\xi)_D}{(q^4 \tilde \xi z^2)_D} \sum_{M'=0}^N c_{M'}(z^2;\bm t) (q^{2(j_0 + M')} \xi \tilde \xi)^{D} 
\]
modulo $\Ker(\Tr_W)$.
Note that the function $(q^2 z^2/\xi)_{j_-}$ depends holomorphically on $z^2$ if $j_- \ge 0$ and has simple poles at 
\[
z^2 = q^{2i} \xi, \qq i \in \{0,1,\ldots, -j_- - 1 \}
\]
if $j_- < 0$.
Since $|j_-|\le N$ and $-N \le j_0 + M' \le 2N$ we obtain bounds for the summation variables and arrive at the claim.
\end{proof}

In the remainder of the paper we will assume $(q, \xi , \tilde \xi) \in S^{(N)}$.
We end this section with the following result.

\begin{lemma} \label{lem:Q:invertible}
For generic $z \in U^{(N)}$, i.e.~all $z$ in the complement of a discrete subset of $U^{(N)}$, $\mc Q(z)$ is invertible.
\end{lemma}

\begin{proof}
Consider the total spin operator
\eq{ \label{Sigma:def}
\Si := \left( \begin{smallmatrix} 1 & 0 \\ 0 & -1 \end{smallmatrix} \right)_1 + \left( \begin{smallmatrix} 1 & 0 \\ 0 & -1 \end{smallmatrix} \right)_2 + \cdots + \left( \begin{smallmatrix} 1 & 0 \\ 0 & -1 \end{smallmatrix} \right)_N \in \End(V^{\ot N}).
}
Given that $|\xi \tilde \xi|< |q|^{2N}$, formulas \eqref{L:def}, \eqref{KW:def} and \eqref{KtW:def} imply
\eq{
\mc T^W(0) = \Tr_W \Big( q^{2 \Si} \xi \tilde \xi \, \Big)^D = \frac{1}{1 - q^{2 \Si} \xi \tilde \xi } .
}
In particular we have $\xi \tilde \xi \notin \{ q^{2N}, q^{2N-4}, \ldots, q^{-2N+4}, q^{-2N} \}$ so that $\mc T^W(0)$ is a well-defined invertible linear map.
Note that $\det \mc T^W(z)$ depends polynomially on the matrix entries of $\mc T^W(z)$.
By Theorem \ref{thm:Q:convergent} each entry depends holomorphically on $z \in U^{(N)}$ and hence $\det \mc T^W(0) \ne 0$ implies $\det \mc T^W(z) \ne 0$ generically.
Therefore $\mc T^W(z)$ is invertible for generic values of $z$, as is $\left( \begin{smallmatrix} z^2 & 0 \\ 0 & 1 \end{smallmatrix} \right)^{\ot N}$.
We obtain the desired conclusion.
\end{proof}


\section{Baxter's TQ relation}

In this section we derive the major result of this paper, namely Baxter's relation for the matrices $\mc T^V(z)$ and $\mc Q(z)$.
We start with proving some commutativity properties of the families $\{ \mc T^V(z) \}_{z\in \C}$ and $\{ \mc T^W(z) \}_{z \in \C}$.

\subsection{Commutativity}

Sklyanin's argument of commuting two-row transfer matrices involves extending the tensor product $V^{\ot N}$, the domain of the operators $\mc T^{V,W}(z)$ by a tensor product of two auxiliary spaces labelled $a$ and $b$.
This auxiliary tensor product is $W \ot V$ in the case of $[\mc T^W(y),\mc T^V(z)]=0$ and $V \ot V$ in the case of $[\mc T^V(y),\mc T^V(z)]=0$.
The proof we give here for $[\mc Q(y),\mc T^V(z)] = 0$ is based on the well-documented one discussed in \cite{Sk88,MN91,FShHY97,Vl15,NR18} but arranged, following Frassek and Sz\'{e}cs\'{e}nyi's approach in \cite{FSz15}, in such a way that partial transpositions are not taken with respect to the infinite-dimensional vector space $W$.

\begin{thrm} \label{thm:commute}
Let $y \in U^{(N)}$ and $z \in \C \backslash \{ \pm q^{-1} y^{-1} \}$.
Then
\eq{\label{QT:intermediate} 
\mc T^W(y) \mc T^V(z) =  \Tr_{a, b} \wt K^V_{b}(z) \wt L_{ab}(yz,1) \wt K^W_a(y,1) \mc M^W_a(y,1) L_{ab}(yz,1) \mc M^V_{b}(z)
}
where in the tensor product $W \ot V \ot V^{\ot N}$ we have labelled the factor $W$ by $a$ and the first factor $V$ by $b$, and written $\Tr_{a,b}$ for the partial trace with respect to $W \ot V$.
Furthermore, for all $y \in U^{(N)}$ and $z \in \C$ we have
\eq{ \label{QT:commute} 
[\mc Q(y),\mc T^V(z)] = 0
}
and for all $y,z \in \C$ we have
\eq{ \label{TT:commute} 
[\mc T^V(y),\mc T^V(z)] = 0.
}
\end{thrm}

\begin{proof}
We will establish \eqref{QT:commute} by proving $[\mc T^W(y),\mc T^V(z)]=0$; in the process we will also prove \eqref{QT:intermediate}.
The proof of \eqref{TT:commute} is analogous to the proof of \eqref{QT:commute} and was already given in \cite{Sk88}.

Let $y \in U^{(N)}$ and $z \in \C \backslash \{ \pm q^{-1} y^{-1} \}$.
Then $L(yz,1)^{\t_b} \in \End(W \ot V)$ is invertible, see \eqref{L:T2}.
We have
\begin{align*}
\mc T^W(y) \mc T^V(z) &= \Tr_a \wt K^W_a(y,1) \mc M^W_a(y,1) \Tr_{b} \mc M^V_{b}(z)^{\t_{b}} \wt K^V_{b}(z)^{\t_{b}} \\
&= \Tr_{a,b} \wt K^W_a(y,1) \mc M^W_a(y,1) \mc M^V_{b}(z)^{\t_{b}}  \wt K^V_{b}(z)^{\t_{b}} \\
&= \Tr_{a,b} \wt K^V_{b}(z)^{\t_{b}} \wt K^W_a(y,1) \mc M^W_a(y,1) \mc M^V_{b}(z)^{\t_{b}}  \\
&= \Tr_{a,b} \wt L_{ab}(yz,1)^{\t_{b}} \wt K^V_{b}(z)^{\t_{b}} \wt K^W_a(y,1) \mc M^W_a(y,1) \mc M^V_{b}(z)^{\t_{b}} L_{ab}(yz,1)^{\t_{b}}
\end{align*}
where we have inserted $L(yz,1)^{\t_b} \wt L(yz,1)^{\t_b} = \Id_{W \ot V}$ and used cyclicity of the trace twice.
Standard properties of the partial transpose yield
\begin{align*} 
\mc T^W(y) \mc T^V(z) &= \Tr_{a,b} \left( \wt K^V_{b}(z) \wt L_{ab}(yz,1) \wt K^W_a(y,1) \right)^{\t_{b}} \left( \mc M^W_a(y,1) L_{ab}(yz,1) \mc M^V_{b}(z) \right)^{\t_{b}}  \\
&= \Tr_{a,b} \wt K^V_{b}(z) \wt L_{ab}(yz,1) \wt K^W_a(y,1) \mc M^W_a(y,1) L_{ab}(yz,1) \mc M^V_{b}(z)
\end{align*}
and we obtain \eqref{QT:intermediate}.
Now additionally assume that $z \ne \pm y$ so that $L(y/z,1) \in \End(W \ot V)$ is invertible.
We insert $L(y/z,1)^{-1} L(y/z,1) = \Id_{W \ot V}$ and obtain
\begin{align*} 
\mc T^W(y) \mc T^V(z) &= \Tr_{a,b} \wt K^V_{b}(z) \wt L_{ab}(yz,1) \wt K^W_a(y,1) L_{ab}(y/z,1)^{-1} L_{ab}(y/z,1) \mc M^W_a(y,1) L_{ab}(yz,1) \mc M^V_{b}(z) \\
&= \Tr_{a,b} L_{ab}(y/z,1)^{-1} \wt K^W_a(y,1) \wt L_{ab}(yz,1) \wt K^V_{b}(z) \mc M^V_{b}(z) L_{ab}(yz,1) \mc M^W_a(y,1) L_{ab}(y/z,1);
\end{align*}
here we have used \eqref{KtW:RE} and \eqref{MW:RE}.
Using cyclicity of the trace and restoring the partial transpositions we obtain
\begin{align*} 
\mc T^W(y) \mc T^V(z) &= \Tr_{a,b} \left( \wt K^W_a(y,1) \wt L_{ab}(yz,1) \wt K^V_{b}(z) \right)^{\t_{b}}  \left( \mc M^V_{b}(z) L_{ab}(yz,1) \mc M^W_a(y,1) \right)^{\t_{b}} \\
&= \Tr_{a,b} \wt K^W_a(y,1) \wt K^V_{b}(z)^{\t_{b}} \wt L_{ab}(yz,1)^{\t_{b}} L_{ab}(yz,1)^{\t_{b}} \mc M^V_{b}(z)^{\t_{b}} \mc M^W_a(y,1).
\end{align*}
Using cyclicity of the trace again, we arrive at
\begin{align*} 
\mc T^W(y) \mc T^V(z) &= \Tr_{a,b} \wt K^W_a(y,1) \wt K^V_{b}(z)^{\t_{b}} \mc M^V_{b}(z)^{\t_{b}} \mc M^W_a(y,1) \\
&= \Tr_{a,b} \wt K^V_{b}(z)^{\t_{b}} \mc M^V_{b}(z)^{\t_{b}} \mc M^W_a(y,1) \wt K^W_a(y,1)   \\
&= \Tr_{b} \wt K^V_{b}(z)^{\t_{b}} \mc M^{V}_{b}(z)^{\t_{b}} \Tr_a \mc M^W_a(y,1) \wt K^W_a(y,1) \\
&= \mc T^V(z) \mc T^W(y)
\end{align*}
as required.
For $y \in U^{(N)}$ and $z \in \C$ the matrix $\mc D(y,z) := [\mc T^W(y), \mc T^V(z) ] \in \End(V^{\ot N})$ is well-defined and depends polynomially on $z$.
For all $y \in U^{(N)}$ we have shown that $\mc D(y, z)=0$ for all but finitely many $z$; hence $\mc D(y,z)=0$ for all $z \in \C$.
This completes the proof.
\end{proof}

The proof of Theorem \ref{thm:commute} cannot easily be modified to show $[\mc Q(y),\mc Q(z)]=0$.
We can define suitable solutions of infinite-dimensional Yang-Baxter equations, namely $R^{++}(y/z) \in \End(W \ot W)$ constructed in \cite[Eq.~(A.3)]{BJMST09} and
\[
R^{+-}(y/z) := (\rho^+_{y,1} \ot \rho^-_{z,1})(\mc{R}) \in \End(W \ot W).
\]
However $R^{+-}(z)^{\t_b}$ is not a well-defined linear operator; some columns of the matrix of $R^{+-}(z)^{\t_b}$ with respect to the ordered basis $(w^0 \ot w^0 , w^0 \ot w^1, \ldots, w^1 \ot w^0 , w^1 \ot w^1, \ldots, \ldots)$ of $W \ot W$ have infinitely many nonzero entries.
Therefore $\wt R^{+-}(z) := ((R^{+-}(z)^{\t_b})^{-1})^{\t_b}$ is ill-defined and we cannot insert the identity $R^{+-}(yz)^{\t_b} \wt R^{+-}(yz)^{\t_b} = \Id_{W \ot W}$.

\subsection{Baxter's TQ relation through trace decomposition in a short exact sequence}

Having established that $\mc T^V(z)$ and $\mc Q(z)$ are well-defined elements of $\End(V^{\ot N})$ for generic values of $z$ which mutually commute, we consider the crucial functional relation satisfied by them.
With our preparations, this is now a simple consequence of the fact that short exact sequences of vector spaces are split.

\begin{lemma} \label{lem:SES}
Suppose we have a short exact sequence of vector spaces
\begin{equation} \label{SES}
\begin{tikzcd}
0 \arrow[r] & A \arrow[r,"\iota"] & B \arrow[r,"\tau"]  & C \arrow[r] & 0 
\end{tikzcd}
\end{equation}
where $B$ is assumed to have a basis.
Then there exist $\iota' \in \End(B,A)$ and $\tau' \in \End(C,B)$ such that $\iota' \circ \iota = \id_A$, $\tau \circ \tau' = \id_C$ and $\iota' \circ \tau' = 0$.
Given $\theta \in \End(B)$ we have
\eq{ \label{splittrace}
\Tr_B \theta = \Tr_A \iota' \circ \theta \circ \iota + \Tr_C \tau \circ \theta \circ \tau'
}
if the traces are well-defined.
Moreover, if there exists $\theta_A \in \End(A)$ such that $\theta \circ \iota = \iota \circ \theta_A$ then
\eq{ \label{splittrace:2}
\Tr_B \theta = \Tr_A \theta_A + \Tr_C \theta_C
}
where $\theta_C := \tau \circ \theta \circ \tau' \in \End(C)$.
\end{lemma}

\begin{proof}
Since this is a short exact sequence of vector spaces and $B$ has a basis, there exists a right inverse for $\tau$ i.e.~$\tau' \in \Hom(C,B)$ such that $\tau \circ \tau' = \id_C$.
By the Splitting Lemma, see for instance \cite[Prop.~3.2]{La02}, we have 
\eq{ \label{B:decomposition}
B = \Ker(\tau) \oplus \Im(\tau') = \Im(\iota) \oplus \Im(\tau') \cong A \oplus C.
}
Now define $\iota' \in \Hom(B,A)$ by $\iota'(\iota(a)) = a$ for all $a \in A$ and $\iota'|_{\Im(\tau')} = 0$.
Automatically we have $\Im(\iota') = A$, $\Im(\tau')=\Ker(\iota')$, $\iota' \circ \iota = \id_A$ and $\iota \circ \iota' + \tau' \circ \tau = \id_B$.
Consider the natural decomposition
\[
\theta = \iota \circ \iota' \circ \theta \circ \iota \circ \iota' + \iota \circ \iota' \circ \theta \circ \tau' \circ \tau + \tau' \circ \tau \circ \theta \circ \iota \circ \iota' + \tau' \circ \tau \circ \theta \circ \tau' \circ \tau.
\]
The second term maps $\Im(\tau')$ to $\Im(\iota)$ and $\Im(\iota)$ to 0 and the third term maps $\Im(\iota)$ to $\Im(\tau')$ and $\Im(\tau')$ to 0.
Also, the first and fourth terms are supported on $\Im(\iota)$ and $\Im(\tau')$, respectively.
Hence 
\[
\Tr_B \theta = \Tr_{\Im(\iota)} \iota \circ \iota' \circ \theta \circ \iota \circ \iota' + \Tr_{\Im(\tau')} \tau' \circ \tau \circ \theta \circ \tau' \circ \tau = \Tr_A \iota' \circ \theta \circ \iota + \Tr_C \tau \circ \theta \circ \tau'
\]
where we have used that $\iota$ is an isomorphism from $A$ to $\iota(A) \subseteq B$ with inverse $\iota'$ and $\tau'$ is an isomorphism from $C$ to $\tau'(C) \subseteq B$ with inverse $\tau$.
This proves \eqref{splittrace}, from which \eqref{splittrace:2} immediately follows.
\end{proof}

We now arrive at the main theorem of the paper.
Recall the coefficient polynomials $p_\pm$ defined in \eqref{p:def}.

\begin{thrm}[Baxter's relation for the open XXZ spin chain]\label{thm:QT:rel}
For generic values of $z$ we have
\eq{ \label{QT:rel}
(1-q^2 z^4) \mc T^V(z) \mc Q(z) = p_+(z) \mc Q(q z) + p_-(z) \mc Q(q^{-1} z).
}
\end{thrm}

\begin{proof}
Recall the definition \eqref{Q:def}.
From \eqref{QT:intermediate} with $y=z$ we obtain
\[
\mc T^W(z) \mc T^V(z) = \Tr_{a,b} \wt K^V_{b}(z) \wt L_{ab}(z^2,1) \wt K^W_a(z,1) \mc M^W_a(z,1) L_{ab}(z^2,1) \mc M^V_{b}(z).
\]
Because of the short exact sequence \eqref{SES:plus} and Lemma \ref{lem:iotatau:row}, both in the case $r=1$, we may apply the formula \eqref{splittrace:2}.
It yields
\[
(1-q^2 z^4) \mc T^W(z) \mc T^V(z) = p_+(z) \Tr_a \wt K^W_a(q z,q) \mc M^W_a(q z,q) + p_-(z)  \Tr_a \wt K^W_a(q^{-1} z,q^{-1}) \mc M^W_a(q^{-1} z,q^{-1}).
\]
By virtue of \eqref{MW:r-dependence} and basic properties of scalar multiples of convergent series, these traces are well-defined.
Hence, by \eqref{QT:totalspin} we have
\begin{align*}
(1-q^2 z^4) \mc Q(z) \mc T^V(z) &= p_+(z) \begin{pmatrix} q^2 z^2 & 0 \\ 0 & 1 \end{pmatrix}^{\ot N} \Tr_a \wt K^W_a(q z,1) \mc M^W_a(q z,1) + \\
& \qq + p_-(z) \begin{pmatrix} q^{-2} z^2 & 0 \\ 0 & 1 \end{pmatrix}^{\ot N}  \Tr_a \wt K^W_a(q^{-1} z,1) \mc M^W_a(q^{-1} z,1) \\
&= p_+(z) \mc Q(qz) + p_-(z) \mc Q(q^{-1} z).
\end{align*}
Now using \eqref{QT:commute} with $y=z$ completes the proof.
\end{proof}

\begin{rmk}
Note that \eqref{SES:plus} is a non-split short exact sequence in the category $U_q(\mfb^+)$-modules, which however descends to a split short exact sequence in the category of complex vector spaces.
Hence, the application of Lemma \ref{lem:SES} in the proof of Theorem \ref{thm:QT:rel} is justified.
\hfill \rmkend
\end{rmk}


\section{Crossing symmetry and Bethe equations} \label{sec:Betheeqns}

In this section we will establish a functional equation for the matrix $\mc Q(z)$ (the so-called crossing symmetry) and derive the Bethe equations.

\subsection{Diagonalizability and crossing symmetry of $\mc T^V(z)$}

\begin{lemma} \label{lem:T:diagonalizable}
For all $z \in \C$ and for generic values of $q$, $\tilde \xi$, $\xi$, $t_1,\ldots,t_N$, the matrices $\mc T^V(z)$ are diagonalizable.
\end{lemma}

\begin{proof}
This can be done analogously to the method outlined in \cite[Sec.~4.1]{Re10} for the closed chain.
Up to an overall scalar, $\mc T^V(z)$ depend holomorphically on each of $q , \xi, \tilde \xi \in \C$ and each $t_n \in \C^\times$.
This is a direct consequence of the definition \eqref{T:def} and properties of its constituents.
By the results in \cite[Sec.~II]{Ka82}, the number of eigenvalues and their algebraic and geometric multiplicities is constant outside a discrete set.
Hence $\mc T^V(z)$ is generically diagonalizable or generically non-diagonalizable.

We will now establish that we are in the former case, as required, by showing this matrix is normal with respect to a suitable inner product, and hence diagonalizable, for uncountably many values of each of the indicated parameters, namely for 
\eq{ \label{parameters:restricted}
(q,\xi,\tilde \xi) \in S^{(N)} \cap \R^3, \qq \bm t \in (\C^\times)^N \text{ such that } |t_1|=\cdots=|t_N|=1.
}

Consider the usual inner product on $V$ defined by $(v^i,v^j)=\del_{ij}$ for $i,j \in \{0,1\}$.
We extend this inner product multiplicatively over tensor factors to define inner products on $V^{\ot N}$ and $V \ot V^{\ot N}$.
With respect to these inner products, the adjoint $X^*$ of a linear map $X$ is given by the conjugate transpose and taking adjoints commutes with taking traces.

The assumptions \eqref{parameters:restricted} on the parameters imply 
\[
R(t_n^{-1} z)^* = R(t_n \wb{z}), \qq K^V(z)^* = K^V(\wb{z}), \qq R(t_n z)^* = R(t_n^{-1} \wb{z}), \qq \wt K^V(z)^* = \wt K^V(\wb{z}).
\]
and hence
\begin{align*}
\mc T^V(z)^* 
&= \Tr_a \bigg( \wt K^V_a(z) R_{a1}(t_1 z) \cdots R_{aN}(t_N z) K^V_a(z) R_{aN}(t_N^{-1}  z) \cdots R_{a1}(t_1^{-1}  z) \bigg)^* \\
&= \Tr_a R_{a1}(t_1 \wb{z}) \cdots R_{aN}(t_N \wb{z}) K^V_a(\wb{z}) R_{aN}(t_N^{-1} \wb{z}) \cdots R_{a1}(t_1^{-1} \wb{z}) \wt K^V_a(\wb{z}) \\
&= \Tr_a \wt K^V_a(\wb{z})R_{a1}(t_1 \wb{z}) \cdots R_{aN}(t_N \wb{z}) K^V_a(\wb{z}) R_{aN}(t_N^{-1} \wb{z}) \cdots R_{a1}(t_1^{-1} \wb{z}) = \mc T^V(\wb{z}).
\end{align*}
By \eqref{TT:commute} we conclude that for uncountably many values of the parameters $q,\xi,\tilde \xi,t_1,\ldots,t_N$ the matrix $\mc T^V(z)$ is normal, as required.
\end{proof}

Compared to transfer matrices of closed spin chains, those of open spin chains typically have an additional symmetry, namely one  of the form $\mc T(pz^{-1}) = (\text{scalar}) \mc T(z)$ for some $p \in \C^\times$.
We will combine this later with Theorem \ref{thm:QT:rel} to deduce an analogous property of $\mc Q(z)$.

\begin{lemma} \label{lem:T:crossing}
We have
\eq{ \label{T:crossing}
\mc T^V(q^{-1} z^{-1}) = (q z^2)^{-2(N+1)} \mc T^V(z).
}
\end{lemma}

\begin{proof}
We make three observations to facilitate the proof.
Abbreviate $\si = \left( \begin{smallmatrix} 0 & - \sqrt{-1} \\ \sqrt{-1} & 0 \end{smallmatrix} \right)$.
Straightforward computations yield the following identities for generic $z$:
\eq{ \label{RK:crossingsymmetry}
\begin{aligned}
\si_1 R(q^{-1} z^{-1}) \si^{-1}_1 &= -q^{-1} z^{-2} R_{21}(z)^{\t_1} && \in \End(V \ot V), \\
\si_1 R_{21}(q^{-1} z^{-1}) \si^{-1}_1 &= -q^{-1} z^{-2} R(z)^{\t_1} && \in \End(V \ot V), \\
\si K^V(q^{-1} z^{-1}) \si^{-1} &= (q^2 z^2 - q^{-2} z^{-2}) \Tr_{2} P_{12} \wt R_{21}(z^2)^{\t} K^V_2(z)^\t && \in \End(V), \\
\si \wt K^V(q^{-1} z^{-1}) \si^{-1} &= q^{-2} z^{-4} (q^2 z^2 - q^{-2} z^{-2})^{-1} \Tr_{2} \wt K^V_2(z)^\t P_{12} R_{12}(z^2)^\t && \in \End(V).
\end{aligned}
}
Furthermore, consider the Yang-Baxter equation \eqref{R:YBE}; by partially transposing with respect to the first tensor factor, left- and right-multiplying by $(R_{13}(z_1/z_3)^{\t_1})^{-1}$ and partially transposing the result with respect to the third tensor factor we obtain the following identity in $\End(V \ot V \ot V)$:
\eq{ \label{R:YBE:t}
R_{12}(z_1/z_2)^{\t_1} \wt R_{13}(z_1/z_3)^\t R_{23}(z_2/z_3)^{\t_3} = R_{23}(z_2/z_3)^{\t_3} \wt R_{13}(z_1/z_3)^\t R_{12}(z_1/z_2)^{\t_1} .
}
Finally, linear algebra in the tensor product $V \ot V \ot V$ and the definition of $\wt R(z)$ imply
\eq{ \label{T:crossing:aux}
\big( \Tr_3 P_{23} R_{12}(z^2)^\t  \wt R_{13}(z^2)^{\t} \big)^{\t_1} = \Tr_3 P_{23} \wt R_{13}(z^2)^{\t_3} R_{12}(z^2)^{\t_2} =  \wt R_{12}(z^2)^{\t_2} \big( \Tr_3 P_{23} \big) R_{12}(z^2)^{\t_2} = \Id.
\hspace{-4pt}
}

Having made these preparatory steps, we note that
\begin{align*}
\mc T^V(q^{-1} z^{-1}) &= \Tr_c \left( \si_c \wt K^V_c(q^{-1} z^{-1}) \si^{-1}_c \right) \left( \si_c R_{1c}(q^{-1} t_1 z^{-1}) \si^{-1}_c \right) \cdots \left( \si_c R_{Nc}(q^{-1} t_N z^{-1}) \si^{-1}_c \right) \cdot \\
& \qq \qq \cdot \left( \si_c K^V_c(q^{-1} z^{-1}) \si^{-1}_c  \right)  \left( \si_c R_{cN}(q^{-1} t_N^{-1} z^{-1}) \si^{-1}_c \right) \cdots \left( \si_c R_{c1}(q^{-1} t_1^{-1} z^{-1}) \si^{-1}_c \right)
\end{align*}
by cyclicity of the trace, where we have given the auxiliary space the label $c$.
Now applying \eqref{RK:crossingsymmetry} (using the labels $b$ and $a$ for the two additional auxiliary spaces, respectively) we obtain
\begin{align*}
\mc T^V(q^{-1} z^{-1}) &= (q z^2)^{-2(N+1)} \Tr_c \Big( \Tr_b \wt K^V_b(z)^{\t_b} P_{bc} R_{cb}(z^2)^\t \Big) R(t_1^{-1} z)_{c1}^{\t_c} \cdots R(t_N^{-1} z)_{cN}^{\t_c} \cdot \\
& \hspace{32mm} \cdot \Big( \Tr_a P_{ac} \wt R_{ac}(z^2)^{\t} K^V_a(z)^{\t_a} \Big) R_{Nc}(t_N z)^{\t_c} \cdots R_{1c}(t_1 z)^{\t_c} \\
 &= (q z^2)^{-2(N+1)} \Tr_{a,b,c} P_{ab} \wt K^V_a(z)^{\t_a} P_{bc} R_{ab}(z^2)^\t R_{a1}(t_1^{-1} z)^{\t_a} \cdots R_{aN}(t_N^{-1} z)^{\t_a} \cdot \\
& \hspace{32mm} \cdot \wt R_{ac}(z^2)^{\t} R_{Nc}(t_N z)^{\t_c} \cdots R_{1c}(t_1 z)^{\t_c} K^V_a(z)^{\t_a},
\end{align*}
where we have combined the traces, pulled the factor $P_{ac}$ all the way to the left and the factor $K^V_a(z)^{\t_a}$ all the way to the right.
Repeatedly using \eqref{R:YBE:t} we arrive at
\begin{align*}
\mc T^V(q^{-1} z^{-1}) &= (q z^2)^{-2(N+1)} \Tr_{a,b,c} P_{ab} \wt K^V_a(z)^{\t_a} P_{bc} R_{ab}(z^2)^\t R_{Nc}(t_N z)^{\t_c} \cdots R_{1c}(t_1 z)^{\t_c} \wt R_{ac}(z^2)^{\t} \cdot \\
& \hspace{32mm} \cdot R_{a1}(t_1^{-1} z)^{\t_a} \cdots R_{aN}(t_N^{-1} z)^{\t_a} K^V_a(z)^{\t_a} \\
&= (q z^2)^{-2(N+1)} \Tr_a R_{Na}(t_N z)^{\t_a} \cdots R_{1a}(t_1 z)^{\t_a} \Big( \! \Tr_b P_{ab} \wt K^V_a(z)^{\t_a} \big( \Tr_c P_{bc} R_{ab}(z^2)^\t \wt R_{ac}(z^2)^{\t} \big) \! \Big) \cdot \\
& \hspace{32mm} \cdot R_{a1}(t_1^{-1} z)^{\t_a} \cdots R_{aN}(t_N^{-1} z)^{\t_a} K^V_a(z)^{\t_a}
\end{align*}
where we have moved the first $N$-fold product of partially transposed R-matrices all the way to the left.
Owing to \eqref{T:crossing:aux} the partial trace over the space labelled $c$ amounts to the identity map and as a consequence so does the trace over the space labelled $b$.
We are left with
\begin{align*}
\mc T^V(q^{-1} z^{-1}) &= (q z^2)^{-2(N+1)} \Tr_a R_{Na}(t_N z)^{\t_a} \cdots R_{1a}(t_1 z)^{\t_a} \wt K^V_a(z)^{\t_a} \cdot \\
& \hspace{32mm} \cdot R_{a1}(t_1^{-1} z)^{\t_a} \cdots R_{aN}(t_N^{-1} z)^{\t_a} K^V_a(z)^{\t_a} \\
&= (q z^2)^{-2(N+1)} \Tr_a \big( \wt K^V_a(z) R_{1a}(t_1 z) \cdots R_{Na}(t_N z) \big)^{\t_a} \big( K^V_a(z) R_{aN}(t_N^{-1} z) \cdots R_{a1}(t_1^{-1} z) \big)^{\t_a}.
\end{align*}
Now \eqref{T:crossing} follows from standard properties of the partial transpose.
\end{proof}

\subsection{The polynomiality and commutativity conjecture for the Q-operator}

The results of the remainder of this section are conditional on the following conjecture.

\begin{conj} \label{conj:Q}
The matrix entries of $\mc Q(z)$ depend polynomially on $z^2$ and for all $y,z\in \C$ we have
\eq{ \label{QQ:commute}
[\mc Q(y),\mc Q(z)] = 0.
}
\end{conj}

In Appendix \ref{sec:Q:diagonalentries} we prove that diagonal entries of $\mc T^W(z)$ depend polynomially on $z^2$.
Moreover, in Appendix \ref{sec:lowN} we show that the conjecture is true for $N=2$.
Conjecture \ref{conj:Q} implies that the set $U^{(N)}$ defined by equation \eqref{U:def} can be replaced by $\C$.\\

\subsection{Diagonalizability and crossing symmetry of $\mc Q(z)$}

\begin{lemma} \label{lem:Q:diagonalizable}
For all $z \in \C$ and generic values of $q$, $\tilde \xi$, $\xi$, $t_1,\ldots,t_N$, the matrices $\mc T^W(z)$ and $\mc Q(z)$ are diagonalizable.
\end{lemma}

\begin{proof}
Note that the claim for $\mc Q(z)$ follows from the claim for $\mc T^W(z)$.
We can follow the same steps as in the proof of Lemma \ref{lem:T:diagonalizable}.
Owing to the definitions \eqref{T:def} and properties of its constituents, as well as Theorem \ref{thm:Q:convergent}, $\mc T^W(z)$ depends holomorphically on each of $q , \xi, \tilde \xi \in \C$ and each $t_n \in \C^\times$, up to an overall factor, and as before $\mc T^W(z)$ is either generically diagonalizable or generically non-diagonalizable.
Again, we will show that $\mc T^W(z)$ is normal and hence diagonalizable if the parameters satisfy \eqref{parameters:restricted}; since these conditions still allows for uncountably many values for each parameter, it follows that $\mc T^W(z)$ is generically diagonalizable.

Note that the commutative ring $\C$ is a $\ast$-ring with the involution given by complex conjugation $z \mapsto \wb{z}$.
The algebra ${\rm osc}_q$ becomes a $\ast$-algebra over this $\ast$-ring if we define
\[
a^* = a^\dag \big(1-q^{2(D+1)}\big)^{-1}, \qq (a^\dag)^*= \big(1-\wb{q}^{2(D+1)}\big) a, \qq f(D)^* = \wb{f(D)}.
\]
The space $W$ is a Hilbert space with respect to the inner product defined by
\[
(w^j,w^k)_W = \del_{jk}.
\]
Moreover, it is easy to check that $X^*$ is the Hermitian adjoint of $X \in \rm osc_q$ with respect to this inner product.

We extend the inner products on $V$ and $W$ multiplicatively over tensor factors to define inner products on $V \ot V^{\ot N}$ and $W \ot V^{\ot N}$, so that taking adjoints commutes with taking traces (provided the traces converge).

The assumptions on the parameters imply
\begin{align*}
L(t_n^{-1}z,1)^* &= \begin{pmatrix} q^D & -t_n \wb{z} q^{D+1} (1-q^{2D})^{-1} a^\dag  \\ - t_n \wb{z} a (1-q^{2D}) q^{-D} & q^{-D}-q^{D+2} t_n^2 \wb{z}^2 \end{pmatrix} \\
& = (q^2)_D^{-1} q^{D(D+2)} L(t_n\wb{z},1) (q^2)_D q^{-D(D+2)} 
\end{align*}
and similarly $L(t_n z,1)^* = (q^2)_D^{-1} q^{D(D+2)} L(t_n^{-1} \wb{z}, 1) (q^2)_D q^{-D(D+2)}$; we also have $K^W(z)^* = K^W(\wb{z})$ and $\wt K^W(z)^* = \wt K^W(\wb{z})$.
We deduce that
\begin{align*}
\mc T^W(z)^* 
&= \Tr_a \bigg( \wt K^W_a(z) L_{a1}(t_1 z) \cdots L_{aN}(t_N z) K^W_a(z) L_{aN}(t_N^{-1}  z) \cdots L_{a1}(t_1^{-1}  z) \bigg)^* \\
&= \Tr_a  (q^2)_D^{-1} q^{D(D+2)} L_{a1}(t_1 \wb{z}) \cdots L_{aN}(t_N \wb{z}) (q^2)_D q^{-D(D+2)} \cdot \\
& \qq \qq \cdot K^W_a(\wb{z})  (q^2)_D^{-1} q^{D(D+2)} L_{aN}(t_N^{-1} \wb{z}) \cdots L_{a1}(t_1^{-1} \wb{z}) (q^2)_D q^{-D(D+2)} \wt K^W_a(\wb{z}) \\
&= \Tr_a \wt K^W_a(\wb{z}) L_{a1}(t_1 \wb{z}) \cdots  L_{aN}(t_N \wb{z}) K^W_a(\wb{z}) L_{aN}(t_N^{-1} \wb{z}) \cdots L_{a1}(t_1^{-1} \wb{z}) = \mc T^W(\wb{z}).
\end{align*}
Now by \eqref{QQ:commute} we conclude that $\mc T^W(z)$ is normal for uncountably many values of the parameters $q,\xi,\tilde \xi,t_1,\ldots,t_N$.
\end{proof}

Using the commutativity of $\{ \mc T^V(z) \}_{z \in \C} \cup \{ \mc T^W(z) \}_{z \in \C} $, we now deduce that this family of operators is simultaneously diagonalizable.
Hence we may restrict Baxter's relation \eqref{QT:rel} and the polynomiality of $\mc T^V(z)$ and $\mc Q(z)$ to joint eigenspaces and these relations descend to statements about eigenvalues.
We summarize this in the following lemma.

\begin{lemma} \label{lem:QT:rel:eigenvalues}
All eigenvalues of $\mc T^V(z)$ and $\mc Q(z)$ are polynomial function of $z^2 \in \C$.
For $z \in \C$, let $T^V(z)$ and $Q(z)$ be simultaneous eigenvalues of $\mc T^V(z)$ and $\mc Q(z)$, respectively.
Then for all $z \in \C$ we have
\eq{ \label{QT:rel:eigenvalues}
(1-q^2 z^4) T^V(z) Q(z) = p_+(z) Q(q z) + p_-(z) Q(q^{-1} z).
}
Hence, if $y \in \C$ is such that $Q(y)=0$ or $y^4 = q^{-2}$ we have
\eq{ \label{QT:rel:eigenvalues:2}
p_+(y) Q(q y) + p_-(y) Q(q^{-1} y) = 0.
}
\end{lemma}

We will relate equation \eqref{QT:rel:eigenvalues:2} to the Bethe equations known from Sklyanin's algebraic Bethe ansatz \cite{Sk88}.
First, we state and prove the promised analogon of Lemma \ref{lem:T:crossing} for $\mc Q(z)$.

\begin{thrm}[Crossing symmetry of Baxter's Q-operator] \label{thm:Q:crossing}
For all $z \in \C^\times$ we have
\[
\mc Q(q^{-1} z^{-1}) = (qz^2)^{-2N} \mc Q(z).
\]
\end{thrm}

\begin{proof}
Define the auxiliary $z$-dependent matrix
\[
\mc X(z) := \frac{  \mc Q(z) \mc Q(z^{-1}) - q^{4N} \mc Q(q^{-1} z) \mc Q(q^{-1} z^{-1})}{ (1 - z^4) \prod_{n=1}^N (1 - t_n^2 z^2)(1 - t_n^{-2} z^2) } .
\]
In the first part of the proof we will show that $\mc X(z) =0$.
By Lemma \ref{lem:T:crossing}, comparing \eqref{QT:rel} as-is to \eqref{QT:rel} with $z \mapsto q^{-1} z^{-1}$ and using \eqref{QT:commute} yields
\eq{ \label{Q:crossing:1}
\begin{aligned}
& - \left( p_+(z) \mc Q(q z) + p_-(z) \mc Q(q^{-1} z) \right) \mc Q(q^{-1} z^{-1}) = \\
& \qq = (q z^2)^{2(N+2)} \mc Q(z) \left( p_+(q^{-1} z^{-1}) \mc Q(z^{-1}) + p_-(q^{-1} z^{-1}) \mc Q(q^{-2} z^{-1}) \right).
\end{aligned}
}
From \eqref{p:def} we derive the functional relations
\begin{gather*}
\frac{p_-(q^{-1} z^{-1})}{p_+(z)} = - q^{2(N-2)} z^{-4(N+2)}, \qq \qq \frac{p_+(q^{-1} z^{-1})}{p_-(z)} = -q^{-2(3N+2)} z^{-4(N+2)}, \\
\frac{p_-(z)}{p_+(z)} = q^{2(N-2)} \frac{\tilde \xi z^2 - 1}{z^2 - q^{-2} \tilde \xi} \frac{\xi z^2 - 1}{z^2 - q^{-2} \xi} \frac{1-q^4 z^4}{1-z^4} \prod_{n=1}^N \frac{1-q^2 t_n^2 z^2}{1-t_n^2 z^2} \frac{1-q^2 t_n^{-2} z^2}{1-t_n^{-2} z^2}
\end{gather*}
and hence \eqref{Q:crossing:1} is equivalent to
\eq{ \label{Q:crossing:2}
\mc X(q z) = q^{-2(N+2)} \frac{\tilde \xi z^2 - 1}{z^2 - q^{-2} \tilde \xi} \frac{\xi z^2 - 1}{z^2 - q^{-2} \xi} \mc X(z).
}
Any matrix entry of $\mc X(z)$ satisfies the same q-difference equation \eqref{Q:crossing:2}.
Also, as a consequence of Conjecture \ref{conj:Q}, the matrix entries of $\mc X(z)$ are rational functions in $z^2$.
Finally, note that the coefficient on the right-hand side of \eqref{Q:crossing:2} depends non-trivially on $z$ because $\tilde \xi \xi \ne q^2$ and $(\xi^2,\tilde \xi^2) \ne (q^2,q^2)$; these inequalities follows straightforwardly from the assumption $|\xi \tilde \xi| < |q|^{2N}$.

Writing an arbitrary rational scalar-valued solution of \eqref{Q:crossing:2} as a quotient of polynomials without common factors, we deduce that it has infinitely many zeroes and hence is the zero function.
It follows that the matrix $\mc X(z)$ is the zero matrix.\\

In analogy with \cite[Eq.~(5.14)]{FSz15}, it is tempting now to define a new auxiliary matrix $\wt{\mc X}(z) = \mc Q(z)^{-1} \mc Q(q^{-1} z)$ and from $\mc X(z)=0$ derive the functional equation $\wt{\mc X}(q^{-1}z^{-1}) = q^{4N} \wt{\mc X}(z)$.
However because of the factor $q^{4N} \ne 1$ this has no solutions in rational functions, so it is necessary to provide an alternative approach.
Namely we consider the auxiliary matrix
\[
\mc Y(z) := (q z^2)^{2N} \mc Q(z)^{-1} \mc Q(q^{-1} z^{-1}),
\]
which is well-defined for generic $z$ by virtue of Lemma \ref{lem:Q:invertible}.
Combining $\mc X(z)=0$ with \eqref{QQ:commute} we obtain the functional relation
\[
\mc Y(q^{-1} z) = \mc Y(z)
\]
for generic $z$.
Again, by descending to the matrix entries and using that the entries of $\mc Y(z)$ depend rationally on $z^2$, we obtain that $\mc Y(z)$ is independent of $z$.
Hence there exists a constant matrix $\mc Y$ such that, for all nonzero $z$,
\eq{ \label{Q:crossing:nearlythere}
\mc Y \mc Q(z) = (q z^2)^{2N}  \mc Q(q^{-1} z^{-1}).
}
Replacing $z$ by $q^{-1}z^{-1}$ for nonzero $z$ we obtain that $\mc Y^2 = \Id$, so that $\mc Y$ is diagonalizable with its spectrum contained in $\{-1,1\}$.
Because $\mc Q(z)$ and $\mc Q(q^{-1}z^{-1})$ are simultaneously diagonalizable it follows that the eigenspaces of $\mc Y$ coincide with the eigenspaces of $\mc Q(z)$ which are independent of $z$.\\

It remains to prove that $\mc Y = \Id$.
If $\mc Q(q^{-1/2})$ is invertible then \eqref{Q:crossing:nearlythere} at $z = q^{-1/2}$ immediately implies the desired result.
Now assume on the contrary that $\mc Q(q^{-1/2})$ is not invertible, i.e.~that it has a zero eigenvalue.
Then $\mc Y = \Id$ follows from the following claim.

\emph{Claim:} if $\mc Y$ has an eigenvalue -1 on an eigenspace of $\mc Q(z)$ with eigenvalue $Q(z)$ such that $Q(q^{-1/2}) = 0$ and if
\eq{ \label{parameters:assumption}
\xi, \tilde \xi \notin q^{2 \Z_{\ge 0}+1}, \qq t_n^2 \notin q^{2 \Z + 1} \text{ for } n \in \{1,2,\ldots,N\}.
}
then we have $Q(q^{k-1/2})=0$ for all $k \in \Z_{\ge 0}$.

Indeed, since $Q(z)$ depends polynomially on $z$, the claim implies that $Q(z)=0$ for all $z \in \C$ which contradicts the generic invertibility of $\mc Q(z)$.
It follows that $\mc Y$ has eigenvalue 1 on all eigenspaces of $\mc Q(z)$ whose eigenvalue vanishes at $z=q^{-1/2}$, for generic values of $\xi$, $\tilde \xi$ and each $t_n$.
Hence $\mc Y = \Id$ for generic values of $\xi$, $\tilde \xi$ and each $t_n$.
Because $\mc Q(z)$ depends analytically on $\xi$ and $\tilde \xi$ and Laurent polynomially on each $t_n$, $\mc Y$ depends meromorphically on these parameters.
Hence $\mc Y = \Id$ for all parameter values which proves the theorem.

In order to prove the claim, note that \eqref{p:def} and \eqref{parameters:assumption} together imply
\eq{ \label{alphaplus:zero}
p_+(q^{k-1/2}) \ne 0  \text{ for all } k \in \Z_{\ge 0}.
}
We prove the claim by induction with respect to $k$.
The assumption $Q(q^{-1/2})=0$ proves the case $k=0$.
Now \eqref{QT:rel:eigenvalues:2} with $y = q^{-1/2}$ gives
\[
p_+(q^{-1/2})Q(q^{1/2}) + p_-(q^{-1/2})Q(q^{-3/2}) = 0.
\]
On the other hand, from \eqref{Q:crossing:nearlythere} we obtain
\[
-Q(z) = (qz^2)^{2N} Q(q^{-1}z^{-1})
\]
so that $p_+(q^{-1/2})Q(q^{1/2})  = p_-(q^{-1/2})Q(q^{-3/2})$ and it follows that $p_+(q^{-1/2})Q(q^{1/2}) = 0$.
By virtue of \eqref{alphaplus:zero} it follows that $Q(q^{1/2})=0$, which proves the case $k=1$.

Now assume the claim is true for $k-1$ and $k$ for some fixed $k \in \Z_{\ge 2}$ and take \eqref{QT:rel:eigenvalues:2} with $y = q^{k-1/2}$.
Since $Q(q^{k-1/2})=Q(q^{k-3/2})=0$ it follows that $p_+(q^{k-1/2})Q(q^{k+1/2}) = 0$ and now \eqref{alphaplus:zero} proves the claim for $k+1$.
This completes the induction step.
\end{proof}

\subsection{Bethe equations}

The following technical result together with the eigenvalue version of Baxter's relation \eqref{QT:rel:eigenvalues:2} will allow us to derive the Bethe equations.

\begin{prop}[Zeroes of the eigenvalues of $\mc Q(z)$] \label{prop:Q:eigenvalue:factorized}
Let $Q(z)$ be an eigenvalue of $\mc Q(z)$.
There exist $M \in \{ 0,1, \ldots,N \}$ and $f,y_1,\ldots,y_M \in \C^\times$ such that
\eq{ \label{Q:eigenvalue:factorized}
Q(z) = f z^{2(N-M)} \prod_{j=1}^M (z^2-y_j^2)(z^2-q^{-2}y_j^{-2}).
}
\end{prop}

\begin{proof}
It suffices to prove this for generic values of $z$.
Since the eigenvalues of $\mc Q(z)$ are not generically zero and depend polynomially on $z^2$ we have
\[
Q(z) = f z^{2N'} \prod_{j=1}^{M'} (z^2-y_j^2)
\]
for some $M' , N' \in \Z_{\ge 0}$, $f,y_1,\ldots,y_{M'} \in \C^\times$.
Combining this with Theorem \ref{thm:Q:crossing} we obtain
\[
(q^{-1} z^{-1})^{2N'} \prod_{j=1}^{M'} (q^{-2} z^{-2} -y_j^2) = (q z^2)^{-2N} z^{2N'} \prod_{j=1}^{M'} (z^2-y_j^2).
\]
This is equivalent to the following identity of polynomials:
\[
(-1)^{M'} q^{2(N-N')} \left( \prod_{j=1}^{M'} y_j^2 \right) z^{4N}  \prod_{j=1}^{M'} (z^2 - q^{-2} y_j^{-2}) = 
z^{2(M'+2N')} \prod_{j=1}^{M'} (z^2-y_j^2).
\]
The powers of $z$ yield that $M' = 2M$ is even and $N' = N-M$.
Therefore $M \le N$ and
\eq{ \label{Q:eigenvalue:factorized:1}
\prod_{j=1}^{2M} (z^2 - q^{-2} y_j^{-2}) = \prod_{j=1}^{2M} (z^2-y_j^2), \qq \qq \prod_{j=1}^{2M} y_j^2 = q^{-2M}.
}
The first equation of \eqref{Q:eigenvalue:factorized:1} requires that the multiset $\{ y_j^2 \}_{j=1}^{2M}$ is invariant under the involution $\psi$ of $\C^\times$ defined by $\psi(Y) = q^{-2} Y^{-1}$.
Combined with the second equation and the fact that $|q|<1$ it follows that the multiplicities of the two $\psi$-fixed points $\pm q^{-1}$ are even.
It follows that $\{ y_1^2, \ldots, y_{2M}^2 \}$ splits up into nontrivial $\psi$-orbits and an even number of $\psi$-fixed points.
Up to relabelling we may assume $y_{M+j}^2 = q^{-2} y_j^{-2}$ for $j \in \{1,\ldots,M \}$ and we obtain \eqref{Q:eigenvalue:factorized}.
\end{proof}

Given the factorization \eqref{Q:eigenvalue:factorized}, setting $y=y_i$ in \eqref{QT:rel:eigenvalues:2} immediately yields the following result.

\begin{thrm}[Bethe equations] \label{thm:Q:BetheEqns}
We have
\eq{
p_+(y_i) Q(q y_i) + p_-(y_i) Q(q^{-1} y_i) = 0 , \quad \hbox{for all}\quad i\in \{1,\ldots,M\},\label{eq:BEs}
}
where $p_\pm(z)$ and $Q(z)$ are given by equations \eqref{p:def} and \eqref{Q:eigenvalue:factorized}.
\end{thrm}

In Appendix \ref{app:ABA} we develop Skylanin's formulation of the Algebraic Bethe Ansatz for the open XXZ model in the conventions of this paper and show it coincides precisely with \eqref{eq:BEs}.


\section{Discussion}

For the closed XXZ chain, the Q-operator can be viewed as a `fundamental' transfer matrix in the following sense \cite{BLZ97}: 
the infinite-dimensional Verma module transfer matrix $\mc T_\mu^+(z)$ can we written as a product of the form 
\eq{
\mc T_\mu^+(z) \propto \mc Q(q^{(\mu+1)/2} z)\overline{\mc Q}(q^{-(\mu+1)/2} z)\label{eq:verma}
}
where $\mu\in\R$ is the highest weight of the Verma module, and $\mc Q(z)$ and $\overline{\mc Q}(z)$ are Q-operators given as traces of the monodromy matrix over two different q-oscillator representations associated with the vector spaces $\displaystyle W=\bigoplus_{j=0}^{\infty} \C w^j $ and $\displaystyle\overline{W}=\bigoplus_{j=-1}^{-\infty} \C w^j$ respectively.
If we choose $\mu\in \Z\geq 0$, then the transfer matrix $\mc T_\mu(z)$ of the $\mu+1$ dimensional representation (which is a quotient of the Verma module of weight $\mu$ by that of weight $-\mu-2$) is given by 
\eq{
\mc T_\mu(z)=\mc T_\mu^+(z)- \mc T_{-\mu-2}^+(z).\label{eq:finiteverma}
}
In particular, choosing $\mu=0$ and using \eqref{eq:verma} and \eqref{eq:finiteverma} gives the quantum-determinant expression involving the operator $\mc Q$. 
The key equation \eqref{eq:verma} can be proven either by writing the tensor product of the two q-oscillator representation in terms of a filtration (see Appendix C of \cite{BJMST09}) or by using the R-matrix factorization method of Derkachov (see \cite{De05,DKK06,De07,BLMS10}).
The connection between these two approaches is explained in \cite{KhTs14}.
The closed Baxter's TQ relation can in turn be proved by combining the quantum determinant expression with the $\mu=1$ case of equation \eqref{eq:finiteverma}.

In this paper we have taken the direct route to the TQ relation that involves the short exact sequence from Lemma \ref{lem:iotatau:intertwine}.
In the second paper in this series we will derive analogues of \eqref{eq:verma} and \eqref{eq:finiteverma} for open chains, thus yielding expressions for the quantum determinant and higher-spin transfer matrices in terms of the Q operator of this paper.
The paper \cite{DM06} deals with a factorization of transfer matrices in terms of Q-operators for both closed and certain open chains (but not the $(\xi,\tilde \xi)$-dependent ones discussed here).

One question is the following: do the results of this paper, in particular the boundary fusion relation of Lemma \ref{KW:iotatau:lem}, extend to non-diagonal K-matrices $K^V(z)$? 
(The general solution to the reflection equation associated to the vector representation of $U_q(\wh{\mfsl}_2)$ was found in \cite{DVGR94}.)
We do not yet have an answer to this question; the main complication is the difficulty of finding a solution $K^W(z)$ of the reflection equation \eqref{KW:RE} in the non-diagonal case.


\appendix

\section{Polynomiality of the diagonal entries of the Q-operator} \label{sec:Q:diagonalentries}

In Appendices \ref{sec:Q:diagonalentries} and \ref{sec:lowN} we use some auxiliary notation.
For $\bm \al = (\al_1,\ldots,\al_N), \bm \be = (\be_1,\ldots,\be_N) \in \{ 0,1\}^N$ and $\bm t \in (\C^\times)^N$ we consider the matrix entries $\mc M(z)^{\bm \al}_{\bm \be} \in {\rm osc}_q$ defined by
\eq{
\label{MW:entries}
\mc M^W(z,1) (w \ot v^{\al_1}\otimes v^{\al_2}\otimes \cdots \otimes  v^{\al_N}) = \sum_{\bm \be \in \{0,1\}^N} \mc M(z)^{\bm \al}_{\bm \be}(w) \ot v^{\be_1} \otimes v^{\be_2} \otimes \cdots \otimes v^{\be_N} 
}
for all $w \in W$, and the matrix entries $\mc T(z)^{\bm \al}_{\bm \be} \in \C$ defined by
\eq{ \label{TW:entries}
\mc T^W(z) (v^{\al_1}\otimes v^{\al_2}\otimes \cdots \otimes  v^{\al_N}) = \sum_{\bm \be \in \{0,1\}^N} \mc T(z)^{\bm \al}_{\bm \be} \, v^{\be_1}\otimes v^{\be_2}\otimes \cdots \otimes v^{\be_N}
}
for $\bm \al = (\al_1,\ldots,\al_N), \bm \be = (\be_1,\ldots,\be_N) \in \{ 0,1\}^N$ and $\bm t \in (\C^\times)^N$.

Moreover, in this section we explicitly write the dependence of $\mc T^W$, $\wt K^W$ and $\mc M^W$ on the parameters $\tilde \xi$ and $\bm t = (t_1,\ldots,t_N)$ where applicable.
In particular, 
\[
\mc T^W(z;\tilde \xi;\bm t) = \Tr_a \wt K^W_a(z,1;\tilde \xi) \mc M^W_a(z,1;\bm t)
\]
and we indicate these parameters also for the matrix entries defined by \eqrefs{MW:entries}{TW:entries}.

Consider the $\Z$-linear map $s: \Z^N \to \Z$ defined by 
\eq{ \label{s:def}
s(\ga_1,\ldots,\ga_N) = \ga_1 + \ldots + \ga_N.
}
As a consequence of the decomposition \eqref{oscq:decomposition} and the commutativity \eqref{QT:totalspin}, for all $\bm t = (t_1,\ldots,t_N) \in (\C^\times)^N$ and $\bm \al, \bm \be \in \{ 0,1\}^N$ there exists $f(z;\bm t;\cdot)^{\bm \al}_{\bm \be} \in \mc F$ such that
\eq{ \label{MW:entries:decomp}
\mc M(z;\bm t)^{\bm \al}_{\bm \be} = \begin{cases} 
a^{s(\bm \be-\bm \al)} f(z;\bm t;D)^{\bm \al}_{\bm \be} & \text{if } s(\bm \al) \le  s(\bm \be) , \\
(a^\dag)^{s(\bm \al - \bm \be)} (q^{2(D+1)})_{s(\bm \al-\bm \be)}^{-1} f(z;\bm t;D)^{\bm \al}_{\bm \be} & \text{if } s(\bm \al) \ge s( \bm \be).
\end{cases} 
}
Note that if $s(\bm \al) = s(\bm \be)$ then the two expressions coincide, as required.

\begin{prop}
For all $\bm \al, \bm \be \in \{ 0,1\}^N$, $\bm t \in (\C^\times)^N$ and $u \in \C^\times$ we have the following recursions:
\eq{ \label{f:recursion} \hspace{-1mm}
\begin{aligned}
f(z;u,\bm t;j)^{0,\bm \al}_{0,\bm \be} &= q^{2j+s(\bm \al-\bm \be)} \big( 
f(z;\bm t;j)^{\bm \al}_{\bm \be} + 
q (q^{-2(j+s(\bm \al-\bm \be))} - 1) z^2 f(z;\bm t;j-1)^{\bm \al}_{\bm \be} 
\big), \\
f(z;u,\bm t;j)_{0,\bm \al}^{1,\bm \be} &= - q^{2j+s(\bm \al-\bm \be)+1} u z \big( 
f(z;\bm t;j)^{\bm \al}_{\bm \be} + 
q (q^{-2(j+s(\bm \al-\bm \be))} u^{-2} - z^2) f(z;\bm t;j-1)^{\bm \al}_{\bm \be}
\big), \hspace{-3mm} \\
f(z;u,\bm t;j)^{1,\bm \al}_{0,\bm \be} &= - q^{2j+s(\bm \al-\bm \be)} u^{-1} z \big( 
(q^{-2j} - q^2) f(z;\bm t;j+1)^{\bm \al}_{\bm \be}+ \\
& \hspace{27mm} + q^{-1} (q^{-2(j+s(\bm \al-\bm \be))} - q^2) (q^{-2j} u^2 - q^2 z^2) f(z;\bm t;j)^{\bm \al}_{\bm \be} \big), \hspace{-3mm} \\
f(z;u,\bm t;j)^{1,\bm \al}_{1,\bm \be} &= q^{2j+s(\bm \al - \bm \be)+1}  \big( 
(q^{-2j} - q^2) z^2  f(z;\bm t;j+1)^{\bm \al}_{\bm \be} + \\
& \hspace{27mm} + q^{-1} (q^{-2(j+s(\bm \al-\bm \be))} u^{-2} - q^2 z^2) (q^{-2j} u^2 - q^2 z^2) f(z;\bm t;j)^{\bm \al}_{\bm \be} \big).
\hspace{-3mm} 
\end{aligned}
}
where $f(z;u,\bm t;j)^{\ga,\bm \al}_{\del,\bm \be}$ is shorthand for $f(z;u,t_1,\ldots,t_N;j)^{\ga,\al_1,\ldots,\al_N}_{\del,\be_1,\ldots,\be_N}$.
\end{prop}

\begin{proof}
Recall the entries $L(z)^\ga_\del$ of the L-operator defined in \eqref{L:entries}.
The main observation is that as a consequence of \eqref{MW:def}, for all $\bm t \in (\C^\times)^N$, $u \in \C^\times$, $\bm \al,\bm \be \in \{0,1\}^N$ and $\ga,\del \in \{0,1\}$ we have
\eq{
\mc M(z;u,\bm t)^{\ga,\bm \al}_{\del,\bm \be} = \sum_\veps L(u z)^\veps_\del \mc M(z;\bm t)^{\bm \al}_{\bm \be} L(\tfrac{z}{u})^\ga_\veps.
}
Since $\ga-\del \in \{-1,0,1\}$ we have 
\eq{ \label{f:recursion:2}
\hspace{-1mm} 
\begin{aligned}
& a^{s(\bm \be-\bm \al)+\del-\ga} f(z;u,\bm t;j)^{\ga,\bm \al}_{\del,\bm \be} = \\
&= \sum_\veps L(u z)^\veps_\del a^{s(\bm \be-\bm \al)} f(z;\bm t;j)^{\bm \al}_{\bm \be} L(\tfrac{z}{u})^\ga_\veps \hspace{38mm} \text{if } s(\bm \al-\bm \be)  \le \min(0,\del - \ga), \\
& (a^\dag)^{s(\bm \al-\bm \be)+\ga-\del} (q^{2(j+1)})_{s(\bm \al-\bm \be)+\ga-\del}^{-1} f(z;u,\bm t;j)^{\ga,\bm \al}_{\del,\bm \be} = \\
&= \sum_\veps L(u z)^\veps_\del (a^\dag)^{s(\bm \al-\bm \be)} (q^{2(j+1)})_{s(\bm \al-\bm \be)}^{-1} f(z;\bm t;j)^{\bm \al}_{\bm \be} L(\tfrac{z}{u})^\ga_\veps  \qq \text{if } s(\bm \al-\bm \be) \ge \max(0 , \del-\ga).
\end{aligned}
\hspace{-6mm} 
}
Because of the choice of combinatorial factor in \eqref{MW:entries:decomp} in the case $s(\bm \al - \bm \be) \ge 0$, the system \eqref{f:recursion:2} can be treated as one equation.
It directly leads to \eqref{f:recursion}.
\end{proof}

Together with the boundary values
\eq{ \label{f:initial}
\begin{aligned}
f(z;\emptyset;j)_\emptyset^\emptyset &= q^{-j^2} (-\xi)^j (q^2z^2/\xi)_j \\
f(z;\bm t;j)^{\bm \be}_{\bm \al} &= 0 \qq \qq \text{if } |s(\bm \al-\bm \be)|>N.
\end{aligned}
}
the system \eqref{f:recursion} defines $f(z;\bm t;\cdot)_{\bm \al}^{\bm \be} \in \mc F$ uniquely.
We will need \eqref{f:recursion} only for $\bm \al = \bm \be$ and $\ga = \del$, in which case it simplifies to 
\eq{
\label{f:recursion:simplified}
\begin{aligned}
f(z;u,\bm t;j)^{0,\bm \al}_{0,\bm \al} &= q^{2j} \big( f(z;\bm t;j)^{\bm \al}_{\bm \al} + q (q^{-2j} - 1) z^2 f(z;\bm t;j-1)^{\bm \al}_{\bm \al} \big), \\
f(z;u,\bm t;j)^{1,\bm \al}_{1,\bm \al} &= q^{2j+1} \big( (q^{-2j} - q^2) z^2  f(z;\bm t;j+1)^{\bm \al}_{\bm \al} + \\
& \hspace{27mm} + q^{-1} (q^{-2j} u^{-2} - q^2 z^2) (q^{-2j} u^2 - q^2 z^2) f(z;\bm t;j)^{\bm \al}_{\bm \al} \big).\hspace{-3mm} 
\end{aligned}
}

As a consequence of \eqref{MW:entries:decomp} we have
\eq{ \label{T:matrixentry}
\mc T(z;\tilde \xi;\bm t)^{\bm \be}_{\bm \al} 
= \begin{cases} \displaystyle
\frac{1}{1 - q^2 \tilde \xi z^2} \sum_{j=0}^\infty q^{j^2} (-\tilde \xi)^j (q^4 \tilde \xi z^2)_j^{-1} f(z;\bm t;j)^{\bm \al}_{\bm \be}
& \text{if } s(\bm \al) = s(\bm \be) \\ 0 & \text{otherwise} .
\end{cases}
}
The key step in showing polynomiality is the following recursion for the entries of $\mc T^W(z)$.

\begin{prop} \label{prop:Q:recursion}
Let $\xi \in \C$, $\tilde \xi \in \C$, $u \in \C^\times$, $\bm t \in (\C^\times)^{N+1}$ and $\bm \al , \bm \be \in \{0,1\}^N$.
We have the following identities:
\begin{align*}
\mc T(z;\tilde \xi;u,\bm t)^{0,\bm \be}_{0,\bm \al} &= \mc T(z;q^{2} \tilde \xi;\bm t)^{\bm \be}_{\bm \al} \\
\mc T(z;\tilde \xi;u,\bm t)^{1,\bm \be}_{1,\bm \al} &=  (q^2 z^2 - \tilde \xi) (1 - \tilde \xi z^2) \mc T(z;q^{-2} \tilde \xi;\bm t)^{\bm \be}_{\bm \al} - q^2 z^2  (\tilde \xi u - u^{-1}) (\tilde \xi^{-1} u - u^{-1}) \mc T(z;\tilde \xi;\bm t)^{\bm \be}_{\bm \al}.
\end{align*}
\end{prop}

\begin{proof}
Without loss of generality we may assume that $s(\bm \al) = s(\bm \be)$.
We will start our derivation of the recursion relations with the recursion \eqref{f:recursion:simplified} and the expression \eqref{T:matrixentry}.
In particular, we have
\[
\mc T(z;\tilde \xi;u,\bm t)^{0,\bm \be}_{0,\bm \al}  = \frac{1}{1 - q^2 \tilde \xi z^2} \sum_{j=0}^\infty q^{j^2} (-\tilde \xi)^j (q^4 \tilde \xi z^2)_j^{-1} \Big( q^{2j} f(z;\bm t;j)^{\bm \be}_{\bm \al} + q (1-q^{2j})z^2 f(z;\bm t;j-1)^{\bm \be}_{\bm \al} \Big).
\]
Shifting the summation variable in the second term leads to
\begin{align*}
\mc T(z;\tilde \xi;u,\bm t)^{0,\bm \be}_{0,\bm \al} 
&= \frac{1}{1 - q^2 \tilde \xi z^2} \sum_{j=0}^\infty q^{j^2} (-q^2 \tilde \xi)^j (q^4 \tilde \xi z^2)_j^{-1} \bigg( 1 - \frac{(1 - q^{2(j+1)}) q^2 \tilde \xi z^2}{1 - q^{2(j+2)} \tilde \xi z^2} \bigg) f(z;\bm t;j)^{\bm \be}_{\bm \al} \\
&= \frac{1}{1 - q^4 \tilde \xi z^2} \sum_{j=0}^\infty q^{j^2} (-q^2 \tilde \xi)^j (q^6 \tilde \xi z^2)_j^{-1} f(z;\bm t;j)^{\bm \be}_{\bm \al},
\end{align*}
as required.
To derive the recursion for $\mc T(z;\tilde \xi;u,\bm t)^{1,\bm \be}_{1,\bm \al}$, we have
\begin{align*}
\mc T(z;\tilde \xi;u,\bm t)^{1,\bm \be}_{1,\bm \al} &= \frac{1}{1 - q^2 \tilde \xi z^2} \sum_{j=0}^\infty q^{j^2} (-\tilde \xi)^j (q^4 \tilde \xi z^2)_j^{-1} \Big( q (1-q^{2(j+1)}) z^2 f(z;\bm t; j+1)^{\bm \be}_{\bm \al} + \\
& \hspace{56mm} + q^{-2j} (u^2  - q^{2(j+1)} z^2) (u^{-2} - q^{2(j+1)} z^2) f(z;\bm t;j)^{\bm \be}_{\bm \al}  \Big) \\
&= \frac{1}{1 - q^2 \tilde \xi z^2} \sum_{j=0}^\infty q^{j^2} (-q^{-2} \tilde \xi)^j (q^4 \tilde \xi z^2)_j^{-1} \cdot \\
& \qu \cdot \Big( q^2 (q^{2(j+1)} z^2 - \tilde \xi^{-1}) (1-q^{2j}) z^2 + (u^2  - q^{2(j+1)} z^2) (u^{-2} - q^{2(j+1)} z^2) \Big) f(z;\bm t; j)^{\bm \be}_{\bm \al} 
\end{align*}
where we have shifted the summation variable in the second term.
Using the identity
\begin{align*}
&  q^2 (q^{2(j+1)} z^2 - \tilde \xi^{-1}) (1-q^{2j}) z^2 + (u^2  - q^{2(j+1)} z^2) (u^{-2} - q^{2(j+1)} z^2)  = \\
& \qq =  (q^{2(j+1)} z^2 - \tilde \xi^{-1}) (q^2 z^2 - \tilde \xi) - q^{2(j+1)} z^2 (\tilde \xi u - u^{-1}) (\tilde \xi^{-1} u - u^{-1})
\end{align*}
we arrive at
\begin{align*}
\mc T(z;\tilde \xi;u,\bm t)^{1,\bm \be}_{1,\bm \al} &= (1 - q^2 \tilde \xi^{-1} z^2) \sum_{j=0}^\infty q^{j^2} (-q^{-2} \tilde \xi)^j (q^2 \tilde \xi z^2)_j^{-1} f(z;\bm t; j)^{\bm \be}_{\bm \al} + \\
& \qq - \frac{q^2 z^2 (\tilde \xi u - u^{-1}) (\tilde \xi^{-1} u - u^{-1}) }{1 - q^2 \tilde \xi z^2} \sum_{j=0}^\infty q^{j^2} (- \tilde \xi)^j (q^4 \tilde \xi z^2)_j^{-1} f(z;\bm t; j)^{\bm \be}_{\bm \al} 
\end{align*}
as required.
\end{proof}

We are now in a position to prove the desired properties of the diagonal entries of $\mc T^W(z)$.

\begin{thrm} \label{thm:Q:poly}
Suppose $\xi , \tilde \xi \in \C$ are such that $|\xi \tilde \xi|<|q|^{2N}$.
Then the diagonal matrix entries of $\mc T^W(z)$ and $\mc Q(z)$ depend polynomially on $z^2$.
\end{thrm}

\begin{proof}
For $N=0$ a direct computation gives
\[
\mc T^W(z) = \Tr \wt K^W(z,1) K^W(z,1) = \frac{1}{1 - q^2 \tilde \xi z^2} \sum_{j \ge 0} \frac{(\xi^{-1} q^2 z^2)_j}{(q^4 \tilde \xi z^2)_j} (\xi \tilde \xi )^j =  \frac{\phi( q^2 z^2/\xi, q^2, q^4 \tilde \xi z^2; \xi \tilde \xi )}{1 - q^2 \tilde \xi z^2} .
\]
We recall the q-Gauss summation formula
\eq{ \label{qGauss}
\phi(a,b,c;\tfrac{c}{ab}) = \frac{(\tfrac{c}{a})_\infty (\tfrac{c}{b})_\infty}{(\tfrac{c}{ab})_\infty(c)_\infty},
}
for $a,b \in \C^\times$ and $c \in \C$ such that $|\tfrac{c}{ab}|<1$, see e.g.~\cite[(1.5.1)]{GR04}, we obtain that $\mc T^W(z) = (1 - \xi \tilde \xi)^{-1}$, clearly polynomial in $z^2$.
Now induction with respect to $N$, using the fact that the coefficients in the recurrences given in Proposition \ref{prop:Q:recursion} depend polynomially on $z^2$, proves the theorem for $\mc T^W(z)$.
By \eqref{Q:def} the corresponding statement for $\mc Q(z)$ follows immediately.
\end{proof}


\section{The case $N=2$} \label{sec:lowN}

In this section we demonstrate for $N=2$ that the matrix entries $\mc T^W(z)$ depend polynomially on $z^2$ and that $[\mc T^W(y),\mc T^W(z)] = 0$.
It then follows from \eqref{Q:def} that the Q-operator has the same properties.
Recall that we assume $|\xi \tilde \xi|<|q|^{2N}=|q|^4$ and $0<|q|<1$.

\subsection{Polynomiality}

Our strategy is straightforward: we use the definitions \eqref{MW:def} and \eqref{T:def} directly.
As a consequence of Theorem \ref{thm:Q:poly}, it suffices to do this for the two non-zero off-diagonal entries $\mc T^W(z)^{0,1}_{1,0}$ and $\mc T^W(z)^{1,0}_{0,1}$.
We have
\begin{align*}
\mc M^W(z)^{0,1}_{1,0} &= L(zt_1)^0_1 L(zt_2)^0_0 K^W(z) L(z/t_2)^1_0 L(z/t_1)^0_0 + L(zt_1)^0_1 L(zt_2)^1_0 K^W(z) L(z/t_2)^1_1 L(z/t_1)^0_0 + \\
& \qq + L(zt_1)^1_1 L(zt_2)^0_0 K^W(z) L(z/t_2)^1_0 L(z/t_1)^0_1 + L(zt_1)^1_1 L(zt_2)^1_0 K^W(z) L(z/t_2)^1_1 L(z/t_1)^0_1   \\
\mc M^W(z)^{1,0}_{0,1} &= L(zt_1)^0_0 L(zt_2)^0_1 K^W(z) L(z/t_2)^0_0 L(z/t_1)^1_0 + L(zt_1)^0_0 L(zt_2)^1_1 K^W(z) L(z/t_2)^0_1 L(z/t_1)^1_0 + \\
& \qq + L(zt_1)^1_0 L(zt_2)^0_1 K^W(z) L(z/t_2)^0_0 L(z/t_1)^1_1 + L(zt_1)^1_0 L(zt_2)^1_1 K^W(z) L(z/t_2)^0_1 L(z/t_1)^1_1 .
\end{align*}
Referring to \eqref{L:def}, \eqref{KW:def} and \eqref{KtW:def} for the explicit formulas of the constituent operators, we present the result of the computations, which make use of \eqref{oscq:rep}:
\begin{align*}
\wt K^W(z) \mc M(z)_{1,0}^{0,1} &= q z^2 \frac{(\xi t_1 - t_1^{-1})(t_2 - \xi t_2^{-1})}{(\xi - z^2)(1-q^2 \tilde \xi z^2)} \times \\
& \qq \times \Bigg( \frac{(z^2/\xi)_D}{(q^4 \tilde \xi z^2)_D} (\xi \tilde \xi)^D - \bigg( 1 - (1-q^2) \frac{\xi - z^2}{\xi - t_1^{-2}} \bigg) \frac{(z^2/\xi)_D}{(q^4 \tilde \xi z^2)_D} (q^2 \xi \tilde \xi)^D \Bigg) \\
\wt K^W(z) \mc M(z)_{0,1}^{1,0} &= q z^2 \frac{(t_1-\xi t_1^{-1})(\xi t_2 - t_2^{-1})}{(\xi - z^2)(1-q^2 \tilde \xi z^2)} \times \\
& \qq \times \Bigg( \frac{(z^2/\xi)_D}{(q^4 \tilde \xi z^2)_D} (\xi \tilde \xi)^D - \bigg( 1 - (1-q^2) \frac{\xi - z^2}{\xi - t_1^2} \bigg) \frac{(z^2/\xi)_D}{(q^4 \tilde \xi z^2)_D} (q^2 \xi \tilde \xi)^D \Bigg).
\end{align*}
Note that the expressions are related to each other via the involution of $(\C^\times)^2$ defined by $(t_1,t_2) \mapsto (t_1^{-1},t_2^{-1})$.
It is therefore sufficient to prove the polynomiality of one of the expressions, after having taken the trace.
We have
\eq{ \label{N=2:intermediate}
\begin{aligned}
\mc T^W(z)_{1,0}^{0,1} &= 
q z^2 \,\Xi(z)\,\frac{(\xi t_1 - t_1^{-1})(t_2 - \xi t_2^{-1})}{(\xi - z^2)(1-q^2 \tilde \xi z^2)} ,\\
\hbox{where} \quad \Xi(z)&:= \phi(z^2/\xi,q^2,q^4 \tilde \xi z^2;\xi \tilde \xi) - \bigg( 1 - (1-q^2) \frac{\xi - z^2}{\xi - t_1^{-2}} \bigg) \phi(z^2/\xi,q^2,q^4 \tilde \xi z^2;q^2 \xi \tilde \xi).
\end{aligned}
}
The basic hypergeometric function in the second terms is q-Gauss summable, but not the first one.
However we can apply one of Heine's contiguous relations, see \cite{He47}, to write it as a linear combination of two q-Gauss summable basic hypergeometric functions.
The relevant relation is
\eq{ \label{contiguous}
\phi(a,b,c;x) = \phi(a,b,q^{-2}c;x) - q^{-2} c x \frac{(1-a)(1-b)}{(1-q^{-2} c)(1-c)} \phi(q^2 a,q^2 b,q^2 c;x)
}
where $a,b,c,x \in \C$ such that $|x|<1$.
In our case it gives 
\begin{align*}
\Xi(z)&=  \phi(z^2/\xi,q^2,q^2 \tilde \xi z^2;\xi \tilde \xi) - q^2 \tilde \xi^2 z^2 \frac{(\xi-z^2)(1-q^2)}{(1-q^2 \tilde \xi z^2)(1-q^4 \tilde \xi z^2)} \phi(q^2 z^2/\xi,q^4,q^6 \tilde \xi z^2;\xi \tilde \xi) + \\
& \qq-  \bigg( 1 - (1-q^2) \frac{\xi - z^2}{\xi - t_1^{-2}} \bigg) \phi(z^2/\xi,q^2,q^4 \tilde \xi z^2;q^2 \xi \tilde \xi).
\end{align*}
Using \eqref{qGauss} we obtain
\begin{align*}
\Xi(z)&= \frac{(q^2 \xi \tilde \xi)_\infty (\tilde \xi z^2)_\infty}{(\xi \tilde \xi)_\infty (q^2 \tilde \xi z^2)_\infty} - q^2 \tilde \xi^2 z^2 (\xi-z^2)(1-q^2) \frac{(q^4 \xi \tilde \xi)_\infty (q^2 \tilde \xi z^2)_\infty}{(\xi \tilde \xi)_\infty (q^2 \tilde \xi z^2)_\infty} + \\
& \qq -  \bigg( 1 - (1-q^2) \frac{\xi - z^2}{\xi - t_1^{-2}} \bigg) \frac{(q^4 \xi \tilde \xi)_\infty (q^2 \tilde \xi z^2)_\infty}{(q^2 \xi \tilde \xi)_\infty (q^4 \tilde \xi z^2)_\infty} \\
&= \frac{1-\tilde \xi z^2}{1- \xi \tilde \xi}  - q^2 \tilde \xi^2 z^2  \frac{(\xi-z^2)(1-q^2)}{(1-\xi \tilde \xi)(1-q^2 \xi \tilde \xi)}  -  \bigg( 1 - (1-q^2) \frac{\xi - z^2}{\xi - t_1^{-2}} \bigg) \frac{1-q^2 \tilde \xi z^2}{1-q^2 \xi \tilde \xi} .
\end{align*}
After collecting terms that are manifestly proportional to $\xi - z^2$ we obtain
\begin{align*}
\Xi(z)&= \frac{1-\tilde \xi z^2}{1- \xi \tilde \xi} - \frac{1-q^2 \tilde \xi z^2}{1-q^2 \xi \tilde \xi} + \frac{(\xi-z^2)(1-q^2)}{1-q^2 \xi \tilde \xi} \bigg( \frac{1-q^2 \tilde \xi z^2}{\xi - t_1^{-2}} -   \frac{q^2 \tilde \xi^2 z^2}{1-\xi \tilde \xi} \bigg).
\end{align*}
The first two terms can be combined, yielding
\begin{align*}
\Xi(z)&= \tilde \xi \frac{(1-q^2)(\xi-z^2)}{(1-\xi \tilde \xi)(1-q^2 \xi \tilde \xi)} + \frac{(\xi-z^2)(1-q^2)}{1-q^2 \xi \tilde \xi} \bigg( \frac{1-q^2 \tilde \xi z^2}{\xi - t_1^{-2}} -   \frac{q^2 \tilde \xi^2 z^2}{1-\xi \tilde \xi} \bigg)\\
&= \frac{(\xi-z^2)(1-q^2)}{1-q^2 \xi \tilde \xi} \bigg( \tilde \xi \frac{1-q^2 \tilde \xi z^2}{1-\xi \tilde \xi} +  \frac{1-q^2 \tilde \xi z^2}{\xi - t_1^{-2}} \bigg).
\end{align*}
Note the appearance of the overall factor $\xi - z^2$.
The two terms inside parentheses can be combined in the same way, yielding
\begin{align*}
\Xi(z)& = (1-q^2) \frac{(t_1-\tilde \xi t_1^{-1})(\xi-z^2)(1-q^2 \tilde \xi z^2)}{(\xi t_1-t_1^{-1})(1-q^2 \xi \tilde \xi)(1- \xi \tilde \xi)},
\end{align*}
so that also an overall factor $1-q^2 \tilde \xi z^2$ has emerged.
Comparing with \eqref{N=2:intermediate}, we see that the two $z$-dependent factors in the denominator are cancelled, yielding the following polynomial in $z^2$: 
\eq{ \label{Q:0110}
\mc T^W(z)^{0,1}_{1,0} =  q z^2 (1-q^2) \frac{(t_1-\tilde \xi t_1^{-1})(t_2 - \xi t_2^{-1})}{(1-q^2 \xi \tilde \xi)(1- \xi \tilde \xi)}.
}

\subsection{Commutativity} 

Considering the locations of the zero entries of $\mc T^W(z)$, we see that the property $[\mc T^W(y),\mc T^W(z)]=0$ is equivalent to the following ratios being independent of $z$:
\eq{ \label{Q:ratiosofentries}
\frac{\mc T^W(z)^{0,1}_{0,1} - \mc T^W(z)^{1,0}_{1,0}}{\mc T^W(z)^{0,1}_{1,0}}, \qq \frac{\mc T^W(z)^{0,1}_{0,1} - \mc T^W(z)^{1,0}_{1,0}}{\mc T^W(z)^{1,0}_{0,1}}.
}
As a consequence of the explicit formula \eqref{Q:0110} and the fact that $\mc T^W(z)^{1,0}_{0,1}$ can be obtained from $\mc T^W(z)^{0,1}_{1,0}$ by inverting $t_1$ and $t_2$, it is sufficient to prove that the first ratio of \eqref{Q:ratiosofentries} is independent of $z$.
To this end, we compute the difference $\mc T^W(z)^{0,1}_{0,1} - \mc T^W(z)^{1,0}_{1,0} \in \C$ following the procedure of the diagonal entries.
We simply present the result of the first part of the computation:
\begin{align*}
& \wt K^W(z) \Big( \mc M^W(z)^{0,1}_{0,1} - \mc M^W(z)^{1,0}_{1,0} \Big) = \\
&= \frac{(1-q^2) z^2 (1-\xi z^2)}{(\xi - z^2)(1-q^2 \tilde \xi z^2)} \frac{(z^2/\xi)_D}{(q^4 \tilde \xi z^2)_D} (\xi \tilde \xi)^D + \\
& \qq + \frac{z^2}{1-q^2 \tilde \xi z^2} \Bigg( q^2 (t_1^2+t_1^{-2}-t_2^2-t_2^{-2}) + (1-q^2) \bigg(q^2 \xi - \frac{1-z^4}{\xi - z^2} \bigg) \Bigg) \frac{(z^2/\xi)_D}{(q^4 \tilde \xi z^2)_D} (q^2 \xi \tilde \xi)^D
\end{align*}
so that
\begin{align*}
& \mc T^W(z)^{0,1}_{0,1} - \mc T^W(z)^{1,0}_{1,0} = \\
&= \frac{(1-q^2) z^2 (1-\xi z^2)}{(\xi - z^2)(1-q^2 \tilde \xi z^2)} \phi(z^2/\xi,q^2,q^4 \tilde \xi z^2;\xi \tilde \xi) + \\
& \qq + \frac{z^2}{1-q^2 \tilde \xi z^2} \Bigg( q^2 (t_1^2+t_1^{-2}-t_2^2-t_2^{-2}) + (1-q^2) \bigg( q^2 \xi - \frac{1-z^4}{\xi - z^2} \bigg) \Bigg) \phi(z^2/\xi,q^2,q^4 \tilde \xi z^2;q^2 \xi \tilde \xi).
\end{align*}
Using the relations \eqref{qGauss} and \eqref{contiguous} this yields
\[
\mc T^W(z)^{0,1}_{0,1} - \mc T^W(z)^{1,0}_{1,0} = q^2 z^2 \bigg( \frac{t_1^2+t_1^{-2}-t_2^2-t_2^{-2}}{1-q^2 \xi \tilde \xi} - \frac{(q-q^{-1})^2(\xi - \tilde \xi)}{(1-\xi \tilde \xi)(1-q^2 \xi \tilde \xi)} \bigg)
\]
and the required property of the ratios \eqref{Q:ratiosofentries} readily follows.


\section{The Algebraic Bethe Ansatz for open chains} \label{app:ABA}

In this section, we summarize the Algebraic Bethe Ansatz (ABA) for open chains \cite{Sk88} in the notation of this paper.
Recall that the objective of the ABA is to find eigenvectors and eigenvalues of the transfer matrix $\mc T^V(z)$ defined by \eqref{T:def} where the monodromy matrix $\mc M^V(z)\in \End(V\otimes V^{\otimes N})$ is defined by \eqref{MV:def}.
It follows from \eqref{KV:RE} that $\mc M^V(z)$ satisfies the reflection equation
\eq{
R(y/z) \mc M^V_1(y) R(y z) \mc M^V_2(z) = \mc M^V_2(z) R(y z) \mc M^V_1(y) R(y/z)\in  \End(V\otimes V\otimes V^{\otimes N}).\label{VV:RE}
}
We may write
\[ 
R(z)=\begin{pmatrix} a(z)&0&0&0\\0&b(z)&c(z)&0\\0&c(z)&b(z)&0\\0&0&0&a(z)\end{pmatrix},\qq 
\mc M^V(z)=\begin{pmatrix} A(z)& B(z)\\C(z)&D(z)\end{pmatrix},
\]
for some $A(z), B(z), C(z), D(z)\in \End(V^{\otimes N})$ and 
\[
a(z):=1-q^2z^2, \qq b(z):=q(1-z^2), \qq c(z):=z(1-q^2).
\] 
Then it follows from taking suitable components of \eqref{VV:RE} that
\begin{align*} 
b(y/z) a(y z) A(z) B(y)&= a(y/z)b(yz)B(y)A(z)-b(y/z)c(yz)  B(z) D(y)  -c(y/z) b(yz)  B(z) A(y)\\
b(z/y)b(yz)D(z) B(y)&= a(z/y) a(yz) B(y) D(z)- c(z/y)a(yz) B(z) D(y) + \\
& \qq -c(z/y) c(yz)A(z) B(y)  +a(z/y) c(yz)A(y) B(z), \\
B(z) B(y) &= B(y) B(z).
\end{align*}
Hence we can write 
\begin{align} 
\label{eq:RTT1} A(z) B(y) &= \alpha_1(z,y) B(y) A(z) + \al_2(z,y) B(z) A(y) \hspace{33.5mm} + \al_4(z,y) B(z) D(y) , \\
\label{eq:RTT2} D(z) B(y) &= \beta_1(z,y) B(y) D(z) + \beta_2(z,y) B(z) D(y) + \beta_3(z,y) B(y) A(z)+\beta_4(z,y)B(z) A(y),
\end{align}
where
\begin{align*}
\alpha_1(z,y)&= \frac{a(y/z) b(yz)}{b(y/z) a(yz) }, \qq  \al_2(z,y)=-\frac{c(y/z) b(yz)}{b(y/z) a(yz) }, \qq  \al_4(z,y)=-\frac{c(yz)}{a(yz) },\\
\beta_1(z,y)&= \frac {  \left(a (yz)^{2}-  c(yz)^{2}\right)a(z/y)}{a(yz) b(yz) b(z/y) },\qq \beta_2(z,y)=-\frac {  \left(a(yz)^{2}-  c(yz)^{2}\right)c(z/y)}{a(yz) b(yz) b(z/y)},
\\
\beta_3(z,y)&= -\frac {c(yz) c(z/y)}{b(y/z) a(yz) b(z/y)^2} 
\left(a(y/z) b(z/y) + b(y/z) a(z/y) \right), 
\\
\beta_4(z,y)&= \frac {c(yz) }{b(y/z) a(yz) b(z/y)^2} 
\left(c(y/z) c(z/y) b(z/y) + b(y/z) a(z/y)^2\right).
\end{align*}
Note that \eqrefs{eq:RTT1}{eq:RTT2} do not have the same number of terms.
Following Sklyanin, we can modify $D(z)$ in order to remove the third term on the right-hand side of \eqref{eq:RTT2}.
Namely, observe that
\[
\frac{\beta_3(z,y)}{\beta_1(z,y)-\alpha_1(z,y)}=\frac{z^2(1-q^2)}{q^2z^4-1} =: f(z)
\]
is independent of $y$ so that we may conisder $\wt{D}(z):=D(z)+f(z) A(z)$.
As a consequence we have
\begin{align}
A(z) B(y) &= \alpha_1(z,y) B(y) A(z) + \wt{\alpha}_2(z,y) B(z) A(y)  + \al_4(z,y) B(z) \wt{D}(y) ,\label{eq:RTT3}\\
\wt{D}(z) B(y) &= \beta_1(z,y) B(y) \wt{D}(z) + \wt{\beta}_2(z,y) B(z) \wt{D}(y) + \wt{\beta}_4(z,y)B(z) A(y),\label{eq:RTT4}
\end{align}
where
\begin{gather*}
\wt{\alpha}_2(z,y):=\al_2(z,y) - \al_4(z,y) f(y), \qq \qq \qq \wt{\beta}_2(z,y):=\beta_2(z,y)+\al_4(z,y) f(z),\\
\wt{\beta}_4(z,y):= \beta_4(z,y)-\beta_2(z,y)f(y)+\wt{\alpha}_2(z,y) f(z).
\end{gather*}

Moving on to finding the eigenvectors and eigenvalues of $\mc T^V(z)$ by making use of the relations \eqref{eq:RTT3} and \eqref{eq:RTT4}, we note that 
\[
\mc T^V(z)= \gamma_+(z) A(z) +\gamma_-(z) \wt{D}(z),
\]
where 
\[
\gamma_+(z) := \widetilde{K}^V(z)^0_0-f(z) \widetilde{K}^V(z)^1_1, \qq \gamma_-(z) := \widetilde{K}^V(z)^1_1.
\]
Let $\Omega^{(N)} := (v^0)^{\otimes N} \in V^{\otimes N}$.
The algebraic Bethe ansatz posits that eigenvectors of $\mc T^V(z)$ are given by \emph{Bethe states}, which are vectors of the form $B(y_1) B(y_2) \cdots B(y_M) \Omega^{(N)}$ for some $M \in \Z_{\ge 0}$ and $(y_1,\ldots,y_M) \in \mathbb{C}^M$.
For $\bm \al \in \{ 0,1\}^N$, recall the notation $s(\bm \al)$ defined in \eqref{s:def}.
Note that \eqref{M:totalspin} implies
\[
B(y) v^{\al_1} \ot \cdots \ot v^{\al_N} \in \bigoplus_{\bm \be \in \{ 0,1\}^N \atop s(\bm \be) = s(\bm \al) +1} v^{\be_1} \ot \cdots \ot v^{\be_N}
\]
and hence, without loss of generality, we may impose $M \le N$.
Moreover, as a consequence of $B(z)B(y_i)=B(y_i)B(z)$ and \eqrefs{eq:RTT3}{eq:RTT4}, we can write
\[
\mc T^V(z) B(y_1) B(y_2) \cdots B(y_M) \Omega^{(N)} \in V^{\otimes N}
\]
as a $\C$-linear combination of $B(y_1) B(y_2) \cdots B(y_M) \Omega^{(N)}$ and the $M$ ``unwanted'' vectors
\eq{
B(z) \prod_{j\neq i} B(y_j)\Omega^{(N)}.\label{eq:unwanted}
}
The Bethe equations are just the conditions on the $y_i$ that guarantee that the Bethe state is an eigenvector of $\mc T^V(z)$, i.e.~the conditions such that the coefficients of the unwanted vectors \eqref{eq:unwanted} all vanish.
In order to derive these equations, consider the coefficient of $B(z) \prod_{j>1} B(y_j) \Omega^{(N)}$ in the expansion of $\mc T^V(z) B(y_1) B(y_2) \cdots B(y_M) \Omega^{(N)}$.
The only contribution to this term arises when moving $\mc T^V(z)$ through the factor $B(y_1)$ and applying \eqrefs{eq:RTT3}{eq:RTT4}.
For $X,Y \in V^{\ot N}$ we write $X \equiv_1 Y$ if and only if $X-Y$ is a linear combination of the ``wanted'' vector $B(y_1) B(y_2) \cdots B(y_M) \Omega^{(N)}$ and the other unwanted vectors $B(z) \prod_{j\neq i} B(y_j)\Omega^{(N)}$ with $i \ne 1$.
We have
\[
\mc T^V(z) B(y_1) B(y_2) \cdots B(y_M) \Omega^{(N)}
\equiv_1 B(z) \left(\phi_+(z,y_1) A(y_1)+\phi_-(z,y_1)\wt{D}(y_1)\right)  B(y_2) \cdots B(y_M) \Omega^{(N)},
\]
where
\begin{align*}
\phi_+(z,y)&:=\gamma_+(z) \wt{\alpha}_2(z,y) + \gamma_-(z) \wt{\beta}_4(z,y) ,\\
\phi_-(z,y)&:=\gamma_+(z) \alpha_4(z,y) + \gamma_-(z) \wt{\beta}_2(z,y) .
\end{align*}
Hence, we have
\begin{align*}
&\mc T^V(z) B(y_1) B(y_2) \cdots B(y_M) \Omega^{(N)} \equiv_1 \\
&\equiv_1  B(z)  B(y_2) \cdots B(y_M)  
\left(\phi_+(z,y_1) \prod\limits_{j\neq 1}  \alpha_1(y_1,y_j)    A(y_1)  +\phi_-(z,y_1)    \prod\limits_{j\neq 1}  \beta_1(y_1,y_j)   \wt{D}(y_1)\right)   \Omega^{(N)}, \\
&\equiv_1 \left(\phi_+(z,y_1)  \Delta_+(y_1) \prod\limits_{j\neq 1}  \alpha_1(y_1,y_j)    +\phi_-(z,y_1)  \Delta_-(y_1)  \prod\limits_{j\neq 1}  \beta_1(y_1,y_j)   \right)   B(z)  B(y_2) \cdots B(y_M) \Omega^{(N)},
\end{align*}
where we have used that $A(y) \Omega^{(N)} = \Delta_+(y)  \Omega^{(N)}$, and $\wt{D}(y) \Omega^{(N)} =\Delta_-(y)  \Omega^{(N)}$ for some $\Delta_\pm(y)$ whose explicit expressions we derive below.
Since $[B(y_i),B(y_j)]=0$, the coefficient of any unwanted vector $B(z) \prod_{j\neq i} B(y_j)\Omega$ can be obtained by swapping $y_1$ and $y_i$.
It follows that the Bethe vector is an eigenvector of $\mc T^V(z)$ if the \emph{Bethe equations} are satisfied: \\[2mm]
\eq{
\phi_+(z,y_i) \Delta_+(y_i)  \prod\limits_{j\neq i}  \alpha_1(y_i,y_j) + \phi_-(z,y_i) \Delta_-(y_i) \prod\limits_{j\neq i}  \beta_1(y_i,y_j)   =0  
\label{eq:BAE}
}
for all $i \in \{1,2,\cdots,M\}$.
In that case, the Bethe state $B(y_1) B(y_2) \cdots B(y_M) \Omega^{(N)}$ is an eigenvector of $\mc T^V(z)$.
Although we do not need this, we mention here that a slightly more careful analysis yields that the corresponding eigenvalue equals
\eq{
\gamma_+(z)  \prod\limits_{j=1}^M  \alpha_1(z,y_j) \Delta_+(z)+   \gamma_-(z)  \prod\limits_{j=1}^M  \beta_1(z,y_j) \Delta_-(z).\label{eq:Teig}
}

Explicit expressions for the eigenvalues $\Delta_\pm(z)$ can be obtained by an inductive argument on the lattice length $N$.
Including the inhomogeneities explicitly in the notation we have
\begin{align*} 
A(z;t_1,\ldots,t_N)\Omega^{(N)} =& a(z/t_N) a(z \,t_N) v^0\otimes A(z;t_1,\ldots,t_{N-1}) \Omega^{(N-1)},\\
D(z;t_1,\ldots,t_N)\Omega^{(N)} =& c(z/t_N) c(z \,t_N) v^0\otimes A(z;t_1,\ldots,t_{N-1}) \Omega^{(N-1)} + \\
& \qq \qq + b(z/t_N) b(z \,t_N) v^0\otimes D(z;t_1,\ldots,t_{N-1}) \Omega^{(N-1)}
\end{align*}
From the identity 
\[
f(z) a(z/y) a(z \,y)  +   c(z/y) c(z \,y) = f(z) b(z/y) b(z \,y)
\]
(a consequence of either the explicit expressions for the various functions or directly the Yang-Baxter equation \eqref{R:YBE})
it follows that
\[
\wt{D}(z;t_1,\ldots,t_N)\Omega^{(N)} = b(z/t_N) b(z \,t_N) v^0\otimes  \wt{D}(z;t_1,\ldots,t_{N-1}) \Omega^{(N-1)}.
\]
For $N=0$ the eigenvalues are given by
\[
A(z;\emptyset) \, = K^V(z)^0_0, \qq \qq  \wt{D}(z;\emptyset) \, = K^V(z)^1_1 + f(z) K^V(z)^0_0 
\]
so that we obtain
\[
\Delta_+(z) = K^V(z)^0_0 \prod_{n=1}^N a(z/t_n) a(z t_n), \qq \Delta_-(z) = \Big( K^V(z)^1_1 + f(z) K^V(z)^0_0 \Big) \prod_{n=1}^N  b(z/t_n) b(z t_n) .
\]
Combining this with the explicit expressions
\eq{ \label{phiAD:explicit}
\begin{aligned}
\phi_+(z,y) &= \frac{(q^2-1) y z (1-q^4 z^4)}{(y^2 - z^2)(1 - q^2 y^2 z^2)} \frac{1-y^4}{1-q^2 y^4} (1- \tilde \xi y^2) \\
\phi_-(z,y) &= \frac{(q^2-1) y z (1-q^4 z^4)}{(y^2 - z^2)(1 - q^2 y^2 z^2)} (q^{-2} \tilde \xi - y^2)
\end{aligned}
}
which follow directly from the above definitions,
we deduce that the Bethe equations \eqref{eq:BAE} take the $z$-independent form
\eq{ \label{eq:BEfinal2}
\hspace{-10pt}
\begin{aligned}
& (1-\txi y_i^2) (1-\xi y_i^2) \left( \prod_{n=1}^N (1-q^2 y_i^2 t_n^2)(1-q^2 y_i^2 t_n^{-2})\right) \prod_{j=1 \atop j \ne i}^M (1-q^{-2}y_i^2y_j^{-2})(1-y_i^2y_j^2) = \\
& = q^{2(N-M)} (\txi -q^2y_i^2)  (\xi-q^2 y_i^2) \left( \prod_{n=1}^N (1-y_i^2t_n^2)(1-y_i^2t_n^{-2}) \right) \prod_{j=1 \atop j \ne i}^M (1-q^2 y_i^2y_j^{-2})(q^{-2}-q^2 y_i^2y_j^2)
\end{aligned}
\hspace{-10pt}
}
for all $i \in \{ 1,2,\ldots,M\}$.
This system is equivalent to \eqref{eq:BEs} owing to the explicit formulae \eqref{p:def} and the factorization \eqref{Q:eigenvalue:factorized}.

 
\section{Baxter's TQ relation for closed chains} \label{app:closed}

We briefly review the construction of the transfer matrix and Baxter's Q-operator for \emph{closed} chains, which generally runs parallel to the theory outlined in this paper but is simpler.
Consider the quantum monodromy matrices
\begin{align*}
\mc M^{V,\rm c}_a(z) &= R_{aN}(\tfrac{z}{t_N}) \cdots R_{a1}(\tfrac{z}{t_1}) && \in \End(V \ot V^{\ot N}), \\
\mc M^{W,\rm c}_a(z,r) &= L_{aN}(\tfrac{z}{t_N},r) \cdots L_{a1}(\tfrac{z}{t_1},r) && \in \End(W \ot V^{\ot N}),
\end{align*}
which are direct analogues of \eqrefs{MV:def}{MW:def}.
These objects satisfy
\begin{align*}
\mc M^{W,\rm c}_a(y,r) \mc M^{V,\rm c}_b(z) &= \big( L_{aN}(\tfrac{y}{t_N},r) R_{bN}(\tfrac{z}{t_N}) \big) \cdots \big( L_{a1}(\tfrac{y}{t_1},r) R_{b1}(\tfrac{z}{t_1}) \big) && \in \End(W \ot V \ot V^{\ot N}) \\
\mc M^{V,\rm c}_b(z)  \mc M^{W,\rm c}_a(y,r) &= \big( R_{bN}(\tfrac{z}{t_N}) L_{aN}(\tfrac{y}{t_N},r) \big) \cdots \big( R_{b1}(\tfrac{z}{t_1}) L_{a1}(\tfrac{y}{t_1},r) \big) && \in \End(W \ot V \ot V^{\ot N})
\end{align*}
and hence
\[
L_{ab}(\tfrac{y}{z},r) \mc M^{W,\rm c}_a(y,r) \mc M^{V,\rm c}_b(z) = \mc M^{V,\rm c}_b(z) \mc M^{W,\rm c}_a(y,r) L_{ab}(\tfrac{y}{z},r),
\]
which is to be compared to Lemma \ref{lem:MW:RE}.
Fix $\ze \in \C^\times$ so that $\ze^D$ is an invertible element of ${\rm osc}_q$.
In analogy with Lemma \ref{lem:iotatau:row}, it follows from \eqref{iotatau:intertwine:plus} and \eqrefs{L:iota}{L:tau} that
\eq{ \label{iotatau:1row}
\begin{aligned}
& \bigg( \ze^D \ot \Big( \begin{smallmatrix} 1 & 0 \\ 0 & \ze \end{smallmatrix} \Big) \ot \Id_{V^{\ot N}} \bigg) \mc M^{W,\rm c}_a(z,r) \mc M^{V,\rm c}_b(z) (\iota(r) \ot \Id_{V^{\ot N}}) = \\
& \qq \qq = p_+^{\rm c}(z) (\iota(r) \ot \Id_{V^{\ot N}}) \Big( \ze^D \ot \Id_{V^{\ot N}} \Big) \mc M^{W,\rm c}_a(qz,qr) \\
& \hspace{40mm} \in \Hom(W \ot V^{\ot N},W \ot V \ot V^{\ot N}), \\
& (\tau(r) \ot \Id_{V^{\ot N}}) \bigg( \ze^D \ot \Big( \begin{smallmatrix} 1 & 0 \\ 0 & \ze \end{smallmatrix} \bigg) \ot \Id_{V^{\ot N}} \Big) \mc M^{W,\rm c}_a(z,r) \mc M^{V,\rm c}_b(z) = \\
& \qq \qq = p_-^{\rm c}(z) \Big( \ze^D \ot \Id_{V^{\ot N}} \Big) \mc M^{W,\rm c}_a(q^{-1}z,q^{-1}r) (\tau(r) \ot \Id_{V^{\ot N}}) \\ 
& \hspace{40mm} \in \Hom(W \ot V \ot V^{\ot N},W \ot V^{\ot N}), 
\end{aligned}
}
where
\[
p_+^{\rm c}(z) = \ze \prod_{n=1}^N (1-t_n^{-2}z^2), \qq \qq p_-^{\rm c}(z) = q^N \prod_{n=1}^N (1-q^2t_n^{-2}z^2).
\]
The explicit expressions for $L(z,r)$ imply the factorization
\eq{ \label{MW:closed:r-dependence}
\mc M^{W,\rm c}(z,r) = \mc M^{W,\rm c}(z,1)  \left( \Id \ot \Big( \begin{smallmatrix} r & 0 \\ 0 & 1 \end{smallmatrix} \Big)^{\ot N} \right) ,
}
cf.~\eqref{MW:r-dependence}.
The transfer matrices are given by
\eq{ \label{Tc:def}
\begin{aligned}
\mc T^{V,\rm c}(z) &= \Tr_a \bigg( \Big( \begin{smallmatrix} 1 & 0 \\ 0 & \ze \end{smallmatrix} \Big) \ot \Id_{V^{\ot N}} \bigg) \mc M^{V,\rm c}_a(z) && \in \End(V^{\ot N}), \\
\mc T^{W,\rm c}(z) &= \Tr_a \Big( \ze^D \ot \Id_{V^{\ot N}} \Big) \mc M^{W,\rm c}_a(z,1) && \in \End(V^{\ot N}).
\end{aligned}.
}
Assuming $|q|<1$ as in the open case, the series implicit in the definition of $\mc T^{W,\rm c}(z)$ converges if 
\eq{ \label{zeta:criterion}
|\ze|<q^N.
}
From the explicit formula \eqref{L:def} it is easy to see that in this case the matrix entries of $\mc T^{W,\rm c}(z)$ are polynomials in $z^2$ of degree at most $N$.
It is well-known that the transfer matrices satisfy
\[
[\mc T^{V,\rm c}(y),\mc T^{V,\rm c}(z)] = [\mc T^{W,\rm c}(y),\mc T^{V,\rm c}(z)] = [\mc T^{W,\rm c}(y),\mc T^{W,\rm c}(z)] = 0;
\]
we refer to \cite{BLZ99} and \cite{BGKNR13} for proofs of these in terms of the universal R-matrix of $U_q(\wh{\mfsl}_2)$.
Moreover, $\mc T^{W,\rm c}(z)$ is generically invertible since $\mc T^{W,\rm c}(0) = (1-\ze q^\Si)^{-1}$ is invertible, where $\Si$ is the total spin operator defined in \eqref{Sigma:def}.

\begin{rmk}
Let $M \in \{ 0,1,\ldots,N \}$ so that $N-2M$ is an eigenvalue of $\Si$.
The integrability objects $\mc T^{V,\rm c}(z)$ and $\mc T^{W,\rm c}(z)$ of the closed chain can be obtained (up to overall scalar factors) as a limit of the corresponding objects for the open chain in a process known as \emph{attenuation}, see \cite[Sec.~2]{LIL14} for $\mc T^{V, \rm c}(z)$ for the details (the argument for $\mc T^{W,\rm c}(z)$ is similar).
In our notation, attenuation corresponds to making replacements and taking a limit as follows:
\eq{ \label{attenuation}
z \mapsto \om z, \qq t_n \mapsto \om t_n \text{ for all } n \in \{1,2,\ldots,N\}, \qq \om \to 0.
}
In particular, the quotients $z/t_n$ are unchanged and the products $zt_n$ vanish.
It leads to the following relation between the twist parameter $\ze$ for the closed chain and the parameters $\xi,\tilde \xi$ for the open chain:
\eq{ \label{zeta:relation}
\zeta = \xi \tilde \xi q^{N-2M}.
}
In this limit we get
\[
\mc T^V(z) \to q^M \, \mc T^{V,\rm c}(z), \qq \mc T^W(z) \to \mc T^{W,\rm c}(z).
\]
Note that $|\xi \tilde \xi|<|q|^{2N}$ and $|q|^N \le |q^{N-2M}| \le |q|^{-N}$ by the convergence criterion required for the proof of Theorem \ref{thm:Q:convergent}, so that $|\xi \tilde \xi q^{N-2M}|<|q|^N$.
In view of \eqref{zeta:relation}, this precisely matches \eqref{zeta:criterion}.
\hfill \rmkend
\end{rmk}

In analogy with \eqref{Q:def}, we define the Q-operator by
\eq{
\label{Qc:def} \mc Q^{\rm c}(z) = \Big( \begin{smallmatrix} z & 0 \\ 0 & 1 \end{smallmatrix} \Big)^{\ot N} \mc T^{W,\rm c}(z) \in \End(V^{\ot N}).
}
By the explicit formula \eqref{L:def}, it follows that the entries of $\mc Q^{\rm c}(z)$ are polynomials in $z$ of degree at most $2N$.
Since $\mc T^{V,\rm c}(z)$ and $\mc T^{W,\rm c}(z)$ commute with operators on $\End(V^{\ot N})$ which are tensorial $N$-th powers of diagonal matrices we have
\[
[\mc Q^{\rm c}(y),\mc T^{V,\rm c}(z)] = [\mc Q^{\rm c}(y),\mc T^{W,\rm c}(z)] = [\mc Q^{\rm c}(y),\mc Q^{\rm c}(z)] = 0.
\]
Baxter's TQ relation now follows as a consequence of \eqrefs{iotatau:1row}{MW:closed:r-dependence} and Lemma \ref{lem:SES}:
\eq{ \label{Qc:TQrelation}
\mc T^{V,\rm c}(z) \mc Q^{\rm c}(z) = p_+^{\rm c}(z) \mc Q^{\rm c}(qz) + p_-^{\rm c}(z) \mc Q^{\rm c}(q^{-1}z).
}

\begin{rmk}
Owing to the additional diagonal prefactors in the definitions of the Q-operators, the open Q-operator \eqref{Q:def} itself does not simplify to the closed Q-operator \eqref{Qc:def} using the straightforward limit \eqref{attenuation}. 
Still, Baxter's relation for the closed chain can be obtained by attenuation, as follows.
First of all, the explicit formula \eqref{Q:def} implies that Baxter's relation for the open chain \eqref{QT:rel} is equivalent to
\eq{ \label{QT:rel:2}
(1-q^2 z^4) \mc T^V(z) \mc T^W(z) = p_+(z)  \mc T^W(q z) \Big( \begin{smallmatrix} q^2 & 0 \\ 0 & 1 \end{smallmatrix} \Big)^{\ot N} + p_-(z) \mc T^W(q^{-1} z) \Big( \begin{smallmatrix} q^{-2} & 0 \\ 0 & 1 \end{smallmatrix} \Big)^{\ot N}. 
}
In the limit \eqref{attenuation} we have
\[
p_+(z) \to \xi \wt \xi \zeta^{-1} p_+^{\rm c}(z), \qq \qq p_-(z) \to q^N p_-^{\rm c}(z).
\]
Let $E(N-2M)$ be the eigenspace of $\Si$ with eigenvalue $N-2M$ for $M \in \{0,1,\ldots,N\}$.
Applying the limit \eqref{attenuation} to \eqref{QT:rel:2}, we obtain
\begin{align*}
&q^M \mc T^{V,\rm c}(z) \mc T^{W,\rm c}(z)|_{E(N-2M)} = \\
&\qu = \Bigg( \xi \wt \xi \ze^{-1} p_+^{\rm c}(z) \Big( \begin{smallmatrix} q^2 & 0 \\ 0 & 1 \end{smallmatrix} \Big)^{\ot N} \mc T^{W,c}(q z) + q^N p_-^{\rm c}(z) \Big( \begin{smallmatrix} q^{-2} & 0 \\ 0 & 1 \end{smallmatrix} \Big)^{\ot N} \mc T^{W,\rm c}(q^{-1} z) \Bigg) \bigg|_{E(N-2M)} \\
&\qu = \Bigg( q^{2N-2M} \xi \wt \xi \ze^{-1} p_+^{\rm c}(z) \mc T^{W,c}(q z) + q^{2M-N} p_-^{\rm c}(z) \mc T^{W,\rm c}(q^{-1} z) \Bigg)\bigg|_{E(N-2M)}
\end{align*}
so that by \eqref{zeta:relation}
\[
\mc T^{V,\rm c}(z) \mc T^{W,\rm c}(z)|_{E(N-2M)} = \Bigg( q^{N-M} p_+^{\rm c}(z) \mc T^{W,c}(q z) + q^{M-N} p_-^{\rm c}(z) \mc T^{W,\rm c}(q^{-1} z) \Bigg)\bigg|_{E(N-2M)}.
\]
Hence
\begin{align*}
\mc T^{V,\rm c}(z) \mc T^{W,\rm c}(z) &= p_+^{\rm c}(z) \Big( \begin{smallmatrix} q & 0 \\ 0 & 1 \end{smallmatrix} \Big)^{\ot N} \mc T^{W,c}(q z) + p_-^{\rm c}(z) \Big( \begin{smallmatrix} q^{-1} & 0 \\ 0 & 1 \end{smallmatrix} \Big)^{\ot N} \mc T^{W,\rm c}(q^{-1} z) 
\end{align*}
which, by the explicit formula \eqref{Qc:def}, is equivalent to \eqref{Qc:TQrelation}. \hfill \rmkend
\end{rmk}

Let $Q^{\rm c}(z)$ be an eigenvalue of $\mc Q^{\rm c}(z)$ so that
\[
Q^{\rm c}(z) = f z^{N-M} \prod_{i=1}^M (z^2 - y_i^2)
\]
for some $f,y_1,\ldots,y_M \in \C^\times$ and $M \in \{ 0,1,\ldots,N\}$.
Since $T^{V,\rm c}(z)$ is well-defined for all $z \in \C$, setting $z=y_i$ in the eigenvalue version of \eqref{Qc:TQrelation} yields the Bethe equations for closed XXZ chains:
\eq{ \label{eq:BEsc}
\begin{gathered}
\left( \prod_{n=1}^N (1-q^2y_i^2t_n^{-2}) \right) \prod_{j=1 \atop j \ne i}^M (q^2 - y_i^2y_j^{-2})  = q^N \ze \left( \prod_{n=1}^N (1-y_i^2t_n^{-2}) \right)   \prod_{j=1 \atop j \ne i}^M (1 - q^2 y_i^2 y_j^{-2}) \\
\hspace{108mm} \text{for all } i \in \{ 1,2,\ldots,M \}.
\end{gathered}
}

\begin{rmk}
Subject to \eqref{zeta:relation}, in the limit \eqref{attenuation} (in particular, $y_i \mapsto \om y_i$ before letting $\om \to 0$) the open Bethe equations \eqref{eq:BEfinal2} tend to \eqref{eq:BEsc}, also see \cite[Rmk.~2.4]{LIL14}.
\hfill \rmkend
\end{rmk}


\end{document}